\documentclass[11pt]{article}
\usepackage{amsmath,amssymb,subfigure,epsfig,bbm,stmaryrd,MnSymbol,mathrsfs,url}
\usepackage{makeidx,overpic}
\makeindex
\newcommand{\balpha}{\boldsymbol{\alpha}}
\newcommand{\bbeta}{\boldsymbol{\beta}}
\newcommand{\supp}{\texttt{supp}}
\newcommand{\bPi}{\boldsymbol{\Pi}}

\usepackage{amsthm}
\newtheorem{theorem}{Theorem}

\usepackage{xcolor}
\usepackage{overpic}
\usepackage{xr}

\newcommand{\R}{\mathbb{R}}

\newcommand{\vct}[1]{\boldsymbol{#1}}

\DeclareMathOperator*{\minimize}{\text{minimize}}

\newcommand{\va}{\vct{a}}
\newcommand{\vb}{\vct{b}}

\newcommand{\ve}{\vct{e}}

\newcommand{\vu}{\vct{u}}

\newcommand{\vx}{\vct{x}}

\begin{document}

\title{Optimal Allocation of Quantized Human Eye Depth Perception for Light Field Display Design}
\author{Alireza Aghasi\footnote{Department of Analytics, Robinson College of Business, Georgia State University, Atlanta, GA 30303}, Barmak Heshmat\footnote{BRELYON, Dept of Optics, 555 Price Avenue, Redwood City , CA, USA 94063}, Leihao Wei\footnote{UCLA, Dept. of Electrical and Computer Eng., 420 Westwood Plaza, Los Angeles, CA, USA 90095},
Moqian Tian\footnote{Meta Co., Dept of Neuroscience, 2855 Campus Dr., San Mateo, CA, USA 94403}, Steven  A. Cholewiak\footnote{Optometry and Vision Science, University of California, Berkeley, CA, USA 94720}}

	\maketitle
	
	\begin{abstract}
		Creating immersive 3D stereoscopic, autostereoscopic, and lightfield experiences are becoming the center point of optical design of future head mounted displays and lightfield displays. However, despite the advancement in 3D and light field displays; there is no consensus on what are the necessary quantized depth levels for such emerging displays at stereoscopic or monocular modalities. Here we start from psychophysical theories and work toward defining and prioritizing quantized levels of depth that would saturate the human depth perception. We propose a general optimization framework, which locates the depth levels in a \emph{globally optimal} way for band limited displays. While the original problem is computationally intractable, we manage to find a tractable reformulation as maximally covering a region of interest with a selection
		of hypographs corresponding to the monocular depth of field profiles. The results show that on average 1731 stereoscopic and 8 monocular depth levels (distributed from 25 cm to infinity) would saturate the visual depth perception. Also the first 3 depth levels should be allocated at (148), then (83, 170), then (53, 90, 170) distances respectively from the face plane to minimize the monocular error in the entire population. The study further discusses the 3D spatial profile of the quantized stereoscopic and monocular depth levels. The study provides fundamental guidelines for designing optimal near eye displays, light-field monitors, and 3D screens. 
	\end{abstract}
	
\section{Introduction}
There has been a lot of traction toward more immersive lightfield and autostereoscopic experiences due to advancement in electronics and micro fabrications. Unlike stereoscopic 3D where the 3D perception is created based only on the left-eye, right-eye disparity with fixed accommodation for eye lens, a richer lightfield experience manipulates the wavefront per eye to create depth perception via accommodation monocularly. This reduces or eliminates the well known accommodation-vergence mismatch \cite{patterson2007human,hoffman2008vergence} and therefore the eye stress \cite{ukai2008visual,vienne2014effect}. There has been a push and many breakthroughs for realizing ever more realistic lightfield experiences in optics, graphics, and display community \cite{huang2017systematic,fattal2013multi,wetzstein2011layered,lanman2011polarization, jones2007rendering, yoo2020optimizing, jo2019tomographic, rolland2000multifocal, matsuda2017focal, tay2008updatable, yaracs2010state, takaki2010multi, teng2015improved}. The major categories of methods for creating such experiences are geometrical \cite{huang2017systematic,fattal2013multi}, computational \cite{wetzstein2011layered,lanman2011polarization, jones2007rendering}, multi-focal \cite{yoo2020optimizing, jo2019tomographic, rolland2000multifocal, matsuda2017focal}, phase based holographic \cite{tay2008updatable, yaracs2010state}, and multi-angular \cite{takaki2010multi, teng2015improved}. Each method has its own weaknesses and advantages. For example, super multi-view (SMV) provides lightfield at very compact form factor but is limited to very small viewing zone (also known as eyebox for head mounted displays) and low resolution \cite{takaki2010multi, teng2015improved}, computational methods increase the resolution but come with haze and temporal flickering artifacts \cite{wetzstein2011layered,lanman2011polarization, jones2007rendering}, and holographic methods mostly struggle with color nonuniformity and fringing or speckle artifacts \cite{tay2008updatable, yaracs2010state}. The multi-focal method provides a clean image but requires large volume and gives a sense of depth discontinuity if its not combined with parallax \cite{yoo2020optimizing, jo2019tomographic, rolland2000multifocal, matsuda2017focal}.

Despite all the research that has been going toward realizing richer lightfield experiences both in industry and academia \cite{patterson2007human,hoffman2008vergence, ukai2008visual,vienne2014effect,huang2017systematic,fattal2013multi,wetzstein2011layered,lanman2011polarization, jones2007rendering, yoo2020optimizing, jo2019tomographic, rolland2000multifocal, matsuda2017focal, tay2008updatable, yaracs2010state, takaki2010multi, teng2015improved}, since the majority of lightfield modalities have been rather recent, there is no clear notion of what is the quantized limit of human eye saturation in depth, and what are even elementary lightfield depth levels that should be created in a discrete quantized manner to provide the richest experience with current band limited solutions \cite{heshmat2018fundamental, yoo2020optimizing}. What is the equivalent digital resolution of the human eye in depth both with binocular and monocular view? Where are these depth levels located in 3D space? Which distances are most important if the system is band limited and cannot render all the depths? These are fundamental questions with major impact on emerging virtual reality (VR) and augmented reality (AR) industry as well as upcoming lightfield displays.

Some of these questions are pondered over for decades in psychophysics of human perception and neuroscience literature with completely different approaches \cite{ogle1932analytical, ogle1956stereoscopic, campbell1957depth, green1965effect, legge1987tolerance, marcos1999depth, granrud1984comparison, anderson2008minus, bulthoff1998top, parker2007binocular, tsushima2014higher, bernal2014depth, ginis2012wide, howard2002seeing,fulvio2017use}. For example, there has been numerous studies on depth perception in specific age groups \cite{granrud1984comparison, anderson2008minus}, stereoscopic depth acuity \cite{ogle1932analytical, ogle1956stereoscopic}, monocular depth of field variations \cite{marcos1999depth, bernal2014depth}, contrast perception \cite{green1965effect}, spatial acuity \cite{campbell1957depth}, accommodation optical characterization \cite{legge1987tolerance, bulthoff1998top, parker2007binocular, tsushima2014higher, bernal2014depth, ginis2012wide, howard2002seeing,fulvio2017use}, and color perception of the eye \cite{hofmann2015advances} and the relations between all these different parameters.

While these studies are very informative, psychophysical studies are human centric with no specific interest in digital augmentation or quantization of human perception with emerging display technologies \cite{fulvio2017use}. For example while there are scattered studies on human eye depth of field for different age groups and pupil conditions \cite{marcos1999depth, granrud1984comparison, anderson2008minus}, there is no accurate or scientific guidelines or technical tool to design a VR or AR display that would satisfy such perception acuity \cite{heshmat2018fundamental, yoo2020optimizing}. A specific example of this gap between these two research communities is the case for eye temporal response, which for a long time \cite{kelly1974spatio} was considered to be around 60 Hz by psychophysics literature, but was rather contradicted by display manufacturers as they realized that the customers, specifically those at the gaming and computer graphics sector, were demanding 120-240Hz frame rates. This contradiction was resolved when new studies found that the eye temporal response is very much color and contrast dependent \cite{davis2015humans}.

In this study we find such design guidelines for depth perception by developing a mathematically efficient optimization framework which uses psychophysical theories. The optimization framework that we will put forth is general and can operate with any set of potential depth of field (DoF) profiles at an arbitrary spatial frequency. Given a dense set of potential DoF profiles, our program allows maximally covering the scope of vision with any $T$ number of profiles, where $T$ is a positive integer. The covering is performed in a way that on average the highest visual quality is achieved. In its conventional form the underlying optimization is intractable, however using ideas from hypographs and shape composition, we manage to present an equivalent formulation with an efficient convex relaxation.

We explore the depth perception quantization to understand what can be the equivalent depth resolution of the eye from display design perspective. We then explore and discuss the priority of these depth levels for a band limited system impacted by user statistics and display parameters. Finally, we briefly discuss the experimental method that may be applied to validate these guidelines. These guidelines can be implemented and mapped into different parameters in a broad range of lightfield displays such as super multiview \cite{huang2017systematic,matsuda2017focal}, high density directional \cite{takaki2006high}, and depth filtering multi-focal \cite{hoffman2008vergence, ukai2008visual, mackenzie2010accommodation, liu2010systematic} lightfield displays both in head mounted or far standing, augmented or virtual reality modalities. They may also be adapted for design of lightfield cameras and eliminating redundant acquisition of depth information \cite{bishop2011light, takahashi2018focal}.

\textbf{Mathematical Notations}: For our technical discussions, we use lowercase and uppercase boldface for vectors and matrices, respectively. Given two vectors $\va, \vb\in\mathbb{R}^n$ with elements $a_i$ and $b_i$, the notation $\va\leq\vb$ indicates that $a_i\leq b_i$, for $i=1,2,\ldots,n$. For a more concise representation, given an integer $n$, we use $[n]$ to denote the set $\{1,2,\ldots,n\}$. Finally, given a function $f:\mathbb{R}^n\to \mathbb{R}$, $\supp~\! f$ denotes the support of the function, that is, $\supp~\! f = \{\vx\in\mathbb{R}^n: f(\vx)\neq 0\}$.

\section{Problem Statement: Monocular Depth Resolution}
There are dozens of stereoscopic and monocular depth cues that contribute to human depth perception \cite{howard2002seeing}, perspective, depth from motion, occlusion, and shadow stereopsis, etc. Here we only focus on optics of two major depth cues that is crystalline lens accommodation and vergence of two eyes and we assume that the rest are perception based and solely neurological. For simplicity and practicality, we only consider these two cues for depth perception as the optical architecture of the display is essential in its fabrication and the content contrast can be varied arbitrarily. For simplicity we frame our mathematical model on monocular depth levels and later on slightly cover stereoscopic depth in the discussion section. The main idea is to quantize and prioritize a discrete level of depths that the display can provide in such a way to minimize the depth perception error for an arbitrary image. For this purpose we first have to have a rough estimate or idea about the maximum or total monocular resolution measured in the lab.

The accommodation capability of individual eyes provides the monocular depth perception. If the largest diopter range reported in the literature for very young eye (15 diopter \cite{granrud1984comparison,anderson2008minus}) is divided by the shallowest depth of field reported in the literature (0.15 D FWHM \cite{marcos1999depth}) then the 100 depth levels are the absolute maximum number that human eye at age of 10 can distinguish with 6 mm pupil size in the dark environment. However, if one assumes an average of 6 D range for adults with 0.15 D depth of field, this maximum number of depth levels reduces rapidly to 40 levels. It is noteworthy to mention that depth of field is a diopter range that depends on many other parameters, image intensity and defocus, wavelength, eye point spread function, and/or MTF at different spatial frequencies which all also varies with eye pupil diameter and age. A common way to measure the DoF at a given eye pupil diameter is to record the image intensity of a delta function and look at the maximum intensities and threshold at a desired number of the normalized profile based on an agreement. The maximum intensities tend to have a gaussian-like profile based on \cite{marcos1999depth}. This type of profile is also noticed for MTF at a given cycles-per-degree. Accommodation range significantly varies with age \cite{anderson2008minus}; therefore, it is essential to consider these two parameters in laying out the physical localizations of these depth levels. Here we use an iterative method to localize the depth levels up to 10 meters. 

The DoF for 2 mm, 4 mm and 6 mm pupil diameter has been previously studied in \cite{marcos1999depth}. Anderson et al. \cite{anderson2008minus} has used the objective method to measure the accommodative amplitude in a wide age range of individuals, and has given a sigmoidal function fit to the measured data. The functions were used to find the max accommodative amplitude. For quantization of monocular depth levels we started at the nearest focal plane given by this max accommodation and iteratively found the next focal plane at a step of DoF from\cite{marcos1999depth}. The iteration stops when the focal plane distance is larger than 10 m (Figure \ref{fig1}(a) and (b)). This figure shows that, as one grows older the distance to the nearest focal plane becomes larger and the number of total focal planes decreases; for example, one is able to distinguish 13 focal planes at age of 10 but can distinguish only 2 quantized focal planes at age of 50.  For depth larger than 10 m the depth levels become exponentially sparse based on this iterative quantization. Additionally, experimental studies confirm that after only 8 m the eye reaches its monocular infinity in all pupil sizes and starts to show accumulative negative error even more severely than stereoscopic depth perception, which will be further discussed in the discussion section \cite{palmisano2010stereoscopic}. If one uses the depth of field profile measurement in  \cite{marcos1999depth} as a convolving profile at a given depth in the iterative method, then a continuum of DoF profiles with distance can be calculated as in Figure \ref{fig1}(c). Here the vertical axis is the normalized maximum intensity profile of a delta function seen at each distance enforced by the iteration steps.  

\begin{figure}[t]
	\centering \begin{overpic}[trim={.5cm 0cm  .5cm .5cm},clip,width=5in]{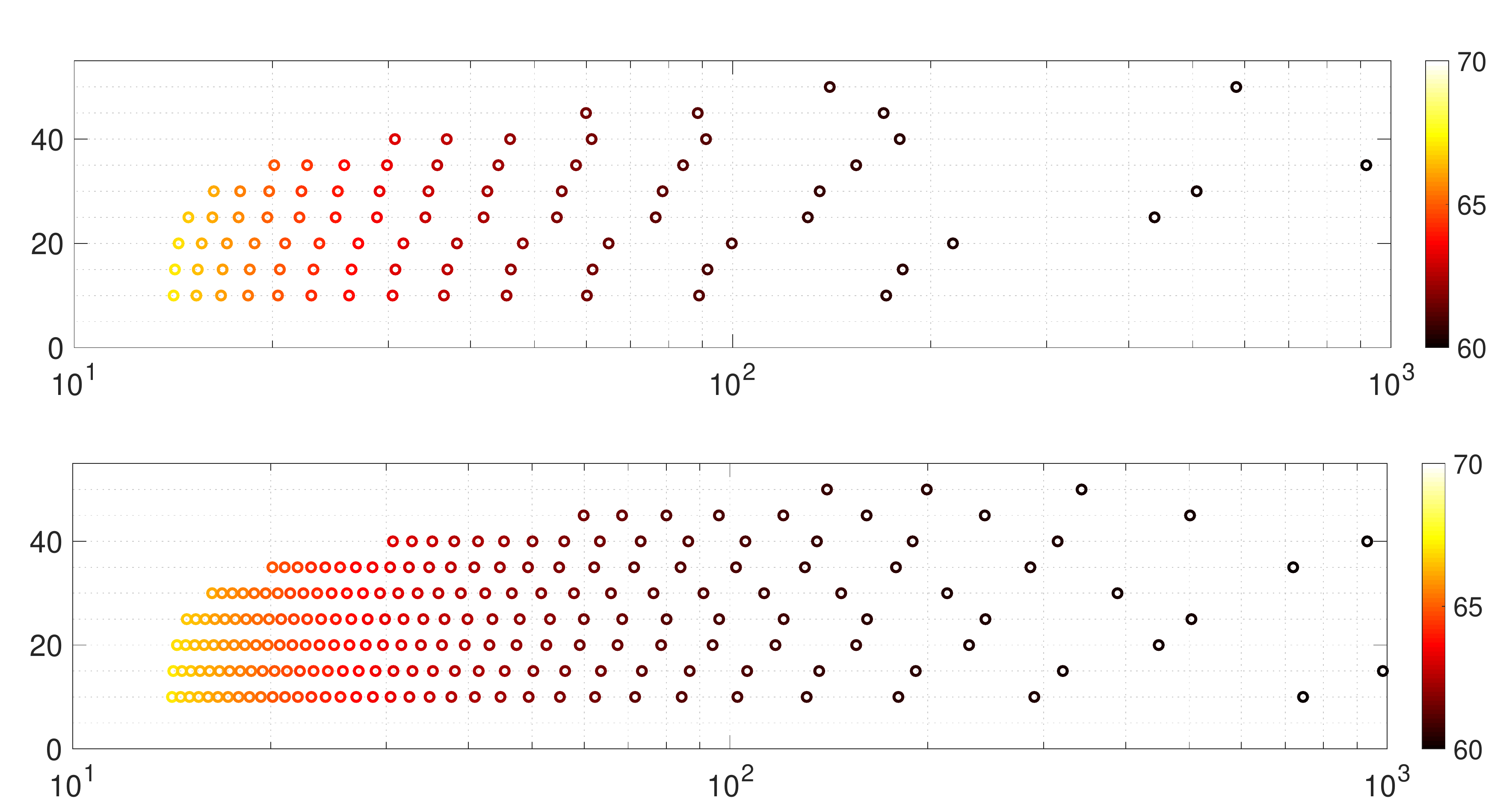}
		\put (43,25) {\scalebox{.8}{(a)}} 
		\put (43,-2.5) {\scalebox{.8}{(b)}} 
		\put (-3,10) {\rotatebox{90}{\scalebox{.7}{Age}}} 
		\put (-3,36) {\rotatebox{90}{\scalebox{.7}{Age}}} 
		\put (52,0) {\scalebox{.7}{Focal plane distance (cm)}} 
		\put (52,27.5) {\scalebox{.7}{Focal plane distance (cm)}} 
		\put (5,48) {\scalebox{.7}{Pupil diameter = 2mm}} 
		\put (5,20.6) {\scalebox{.7}{Pupil diameter = 6mm}} 
	\end{overpic}\\[.25cm]
	\hspace{-.25cm} \begin{overpic}[trim={.5cm 0cm  .3cm .5cm},clip,width=5in]{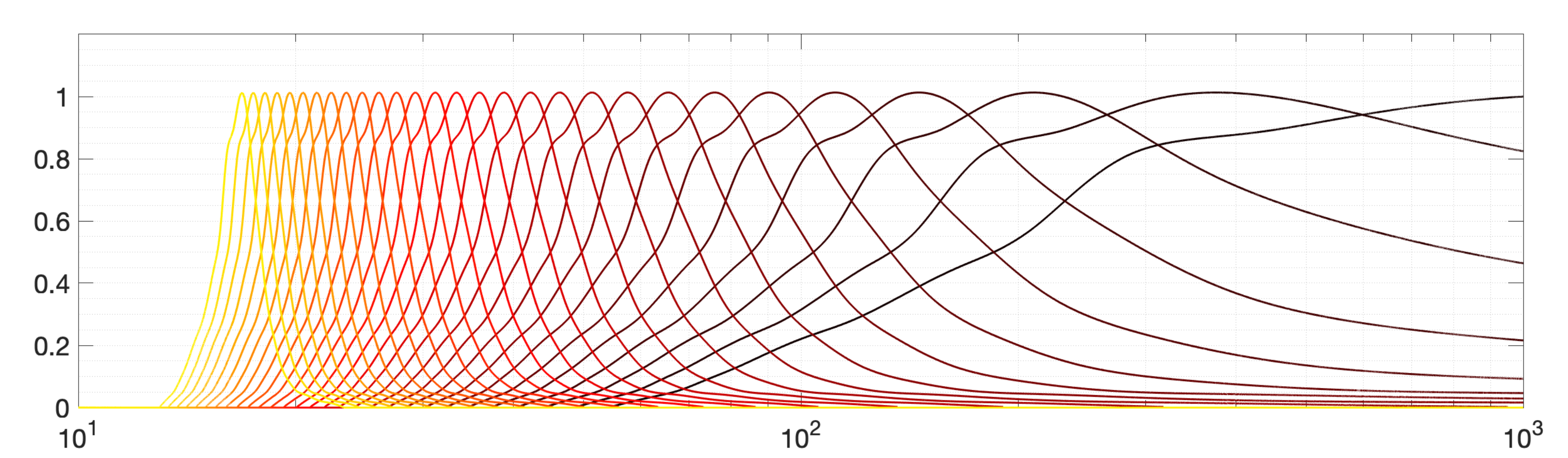}
		\put (43,-3) {\scalebox{.8}{(c)}} 
		\put (52,-1) {\scalebox{.7}{Focal plane distance (cm)}} 
		\put (-3,3) {\rotatebox{90}{\scalebox{.7}{Normalized intensity}}} 
	\end{overpic}
	\vspace{.1cm}
	\caption{ Quantized monocular distinguishable focal planes with age and pupil diameter variations. Focal plane distances are within 10 meters. Colorbar shows the corresponding total eye diopter (60 D from a relaxed eye, plus accommodation power).  (a): Pupil size equals 2 mm. (b): Pupil size is 6 mm. (c): The depth of field profiles for a pupil diameter 6 mm and age 30. }\label{fig1}
\end{figure}
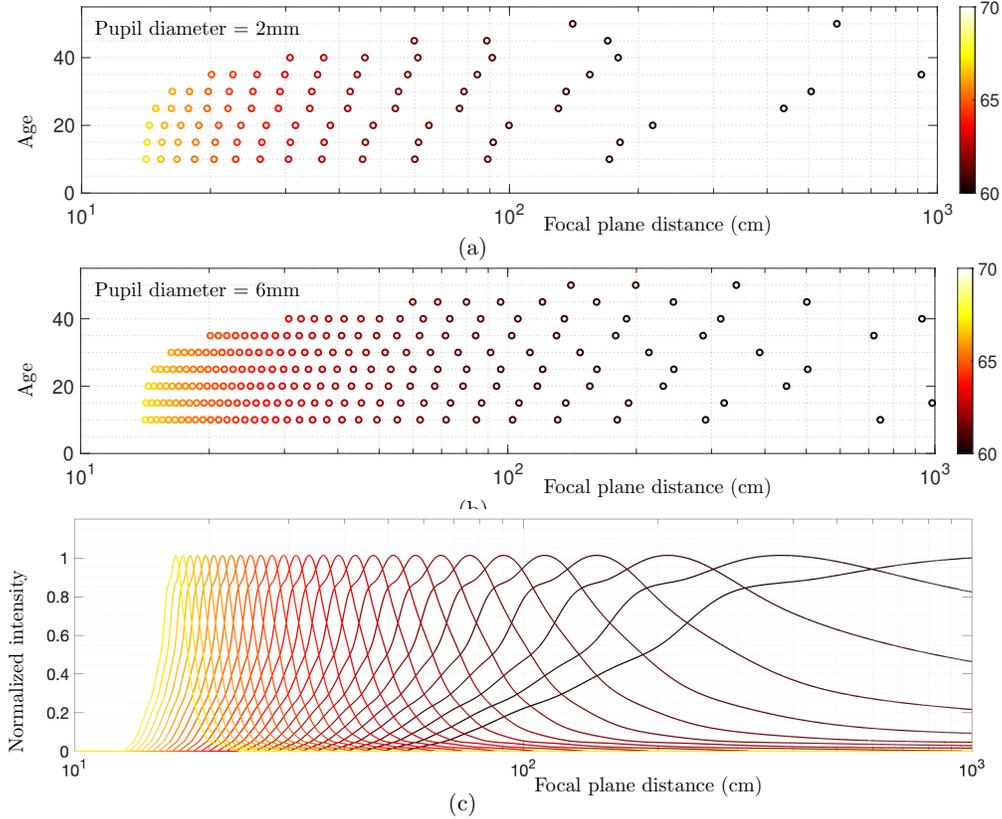

Now that we have found a rough estimate of these monocular depth levels using such iterative approach, given the statistics of eye diopters across different ages \cite{anderson2008minus} and based on daily task operations \cite{Richter2019Always}, we would want to find the priority of allocating depth levels in an optimal fashion. This is significant in the design of 3D displays as the bandwidth is limited and only a limited number of monocular depth levels can be rendered.  The model has to be universally applicable to any type of displays that provides monocular depth; therefore, the methodology that we put forth is general. 

To explain the problem mathematically, in Figure \ref{fig2}(a) we have shown a train of 30 functions (symbolically representing a continuum of DoF profiles), where the goal is selecting $T=4$ functions such that the union of areas under them maximally covers the space (representing the design of a display which restricts the number of depth levels to $T=4$). Figure \ref{fig2}(b) shows the optimal selection of the functions, whereas Figure\ref{fig2}(c) shows an alternative selection which covers a significantly smaller area. As will be discussed in  the next section, this is a challenging combinatorial problem, solving of which would be the focus of next section. 
\begin{figure}[htb!]
	\centering \begin{overpic}[trim={1.8cm 1.2cm  1.5cm 1.2cm},clip,width=5in]{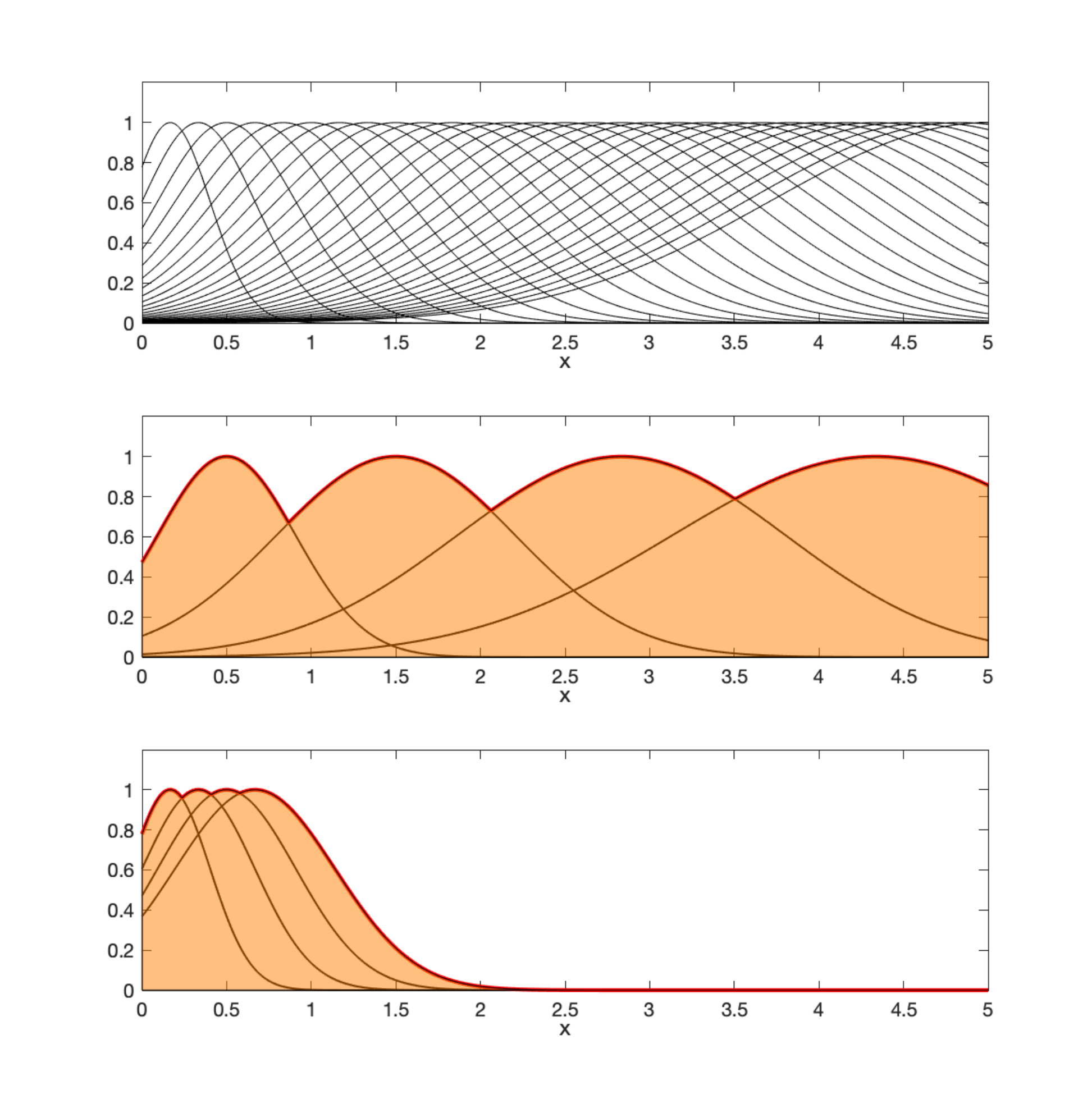}
		\put (48,66.5) {\scalebox{.8}{(a)}} 
		\put (48,32) {\scalebox{.8}{(b)}} 
		\put (48,-2) {\scalebox{.8}{(c)}} 
	\end{overpic}
	\vspace{.1cm}
	\caption{Problem statement.	(a): A set of $n=30$ positive and bounded Gaussian knolls defined in $[0,5]$. (b): A selection of $T=4$ knolls from the functions in (a), which together maximally cover the box defined by $[0,5] \times [0,1]$. (c): For any alternative selection of four knolls other than that shown in panel (b), the union of the areas under the knolls covers a smaller area.}\label{fig2}
\end{figure}

\section{Mathematical Modeling}\label{section:math}
Consider a closed domain $D\in\R^s$ and a set of $n$ positive and bounded functions $f_1,f_2,\ldots,f_n:D\to [0,1]$. For convenience, we adopt the naming convention in \cite{aghasi2013sparse}, and refer to each $f_i$ as a \emph{knoll}. To maintain a general formulation, the only assumptions that we make about a knoll is being bounded (i.e., $0\leq f_i(\vx)\leq 1$) and defined for every point in $D$. Now, consider $T\in\mathbb{N}$, which is a number less than $n$. Restricting the selection to $T$ knolls, the goal is picking a subset of $f_1,f_2,\ldots,f_n$ such that the union of the volumes under these selected knolls maximally covers the box $D\times [0,1]$. To relate the notation, Figure \ref{fig2}(a) shows $n=30$ Gaussian-shaped knolls with varying means and variances defined in $D= [0,5]$, and Figure \ref{fig2}(b) shows the optimal selection of $T=4$ functions that maximally cover the box $[0,5]\times [0,1]$.

The proposed problem can be cast as an optimization program with binary decision variables. More specifically, to find the $T$ optimal covering knolls, one may address the following combinatorial program
\begin{align}\notag
	\minimize_{\alpha_1,\ldots,\alpha_n}&~~\int_D \mbox{d}\vx - \int_D ~\max\{\alpha_1f_1(\vx),\alpha_2 f_2(\vx),\ldots\alpha_n f_n(\vx)\}~\mbox{d}\vx\\ \notag \mbox{subject to:}& ~~\sum_{i=1}^n\alpha_i \leq T\\ 
	& ~~\alpha_i\in\{0,1\}, ~~~i\in[n].\label{eqT1}
\end{align}
Notice that the objective of \eqref{eqT1} is the difference between the area (volume) of the box $D\times[0,1]$ and the selected knolls. Hence minimizing this objective  certifies a maximum covering of the box with the $T$ selected knolls.

Mathematically, the objective in \eqref{eqT1} is concave in $\alpha_i$ (because $\max\{\alpha_1f_1(\vx),\alpha_2 f_2(\vx),\ldots\alpha_n f_n(\vx)\}$ is convex),  which poses an additional challenge on top of having binary decision variables. Usually, when the objective function and  the constraints are convex (disregarding the binary constraints), a relaxation of the binary decision variables yields a convex relaxation to the original combinatorial problem. Such convex relaxation is useful as a computational approximation of the original problem and in many cases can even accurately find the combinatorial solution. Unfortunately, in its original form, the combinatorial problem (\ref{eqT1}) does not offer a straightforward convex relaxation. 

In the sequel we propose a reformulation of \eqref{eqT1}, which remedies the computational issues stated above, and allows us to solve it for the global minimizer.

\subsection{An Alternative Binary Program with Convex Relaxation}
To present a computationally tractable reformulation of \eqref{eqT1}, we proceed by introducing some technical notions. 

Given a function $f(\vx):\R^s\to\R$, the \emph{hypograph} of $f$ is the set of points lying on or below its graph, i.e., 
\[\texttt{hyp}(f) = \left\{ (\vx,y)\in\R^s\times\R: y\leq f(\vx)\right\}\subseteq \R^{s+1}.
\]
The following theorem, relates the point-wise maximum of a set of functions to the union of their corresponding hypographs. 
\begin{theorem}\label{th1}
	Given functions $f_1,f_2,\ldots,f_n:\R^s\to\R$, consider the pointwise maximum-function defined as
	\[f_\lor (\vx) = \max \{f_1(\vx), f_2(\vx),\ldots,f_n(\vx)\}.\]
	Then
	$\texttt{hyp}(f_\lor) = \bigcup_{i=1}^n \texttt{hyp}(f_i).
	$
\end{theorem}		
\begin{proof}
	The proof follows from the basic properties of the pointwise-maximum operation:
	\begin{align*}
		\texttt{hyp}(f_\lor) &= \left\{ (\vx,y)\in\R^s\times\R: y\leq \max \{f_1(\vx), f_2(\vx),\ldots,f_n(\vx)\}\right\}\\&= \left \{ (\vx,y)\in\R^s\times\R: \left(y\leq  f_1(\vx)\right) \lor \left(y\leq  f_2(\vx)\right)\lor \ldots  \left (y\leq  f_n(\vx)\right )\right\}\\&=\bigcup_{i=1}^n  \left\{ (\vx,y)\in\R^s\times\R: y\leq f_i(\vx)\right\}\\&= \bigcup_{i=1}^n \texttt{hyp}(f_i).
	\end{align*}	
	In the second equality, the operand $\lor$ denotes the logical OR operator. 	
\end{proof}	
Using the result of Theorem \ref{th1}, program (\ref{eqT1}) can be cast as finding a set $\Omega\subseteq \{1,2,\ldots,n\}$, which
\begin{align}
	\minimize_{\Omega~\subseteq ~ \{1,2,\ldots,n\}}&~~\int_D\mbox{d}\vx - \int_{\bigcup_{i\in\Omega} \texttt{hyp}(f_i) } \mbox{d}\vx~\mbox{d}y, \quad  \mbox{subject to:} ~~|\Omega| \leq T.\label{eqT2}
\end{align}
Here, $|\Omega|$ denotes the cardinality of the set $\Omega$. One may immediately observe that if $\Omega^*$ is a solution to \eqref{eqT2}, then the set of $\alpha_i^*$ quantities defined via
\[\alpha_i^* = \left\{\begin{array}{cc}1 & i\in\Omega^*\\ 0 & i\notin\Omega^* \end{array} \right., i\in[n], 
\]
forms a solution to \eqref{eqT1}. 

Program (\ref{eqT2}) can be further reshaped into a more standard form.  Consider $\pi_{f_i}(\vx,y)$ to be the indicator function of the set $\texttt{hyp}(f_i)$, i.e.,
\[\pi_{f_i}(\vx,y) = \left\{\begin{array}{lc} 1&(\vx,y)\in \texttt{hyp}(f_i) \\ 0 & (\vx,y)\notin \texttt{hyp}(f_i) \end{array}\right..
\]
It is straightforward to see that 
\[\bigcup_{i\in\Omega} \texttt{hyp}(f_i) = \supp \sum_{i\in\Omega} \pi_{f_i}(\vx,y),
\]
where $\supp$ denotes the support of a given function (the set of points on which the function takes nonzero values). While the functions $\pi_{f_i}(\vx,y)$ are binary-valued, the function $\sum_{i\in\Omega} \pi_{f_i}(\vx,y)$ can take integer values greater than one, depending on the level of overlap among the functions $\pi_{f_i}$. Using the function $c(u) = \min(u,1)$ to clip the values of $\sum_{i\in\Omega} \pi_{f_i}$  that are greater than one,  the objective in \eqref{eqT2} can be related to $\sum_{i\in\Omega} \pi_{f_i}$ via
\[\int_{\bigcup_{i\in\Omega} \texttt{hyp}(f_i) } \mbox{d}\vx~\mbox{d}y = \int_{D\times [0,1]} \min\left( \sum_{i\in\Omega} \pi_{f_i}(\vx,y),1\right)\mbox{d}\vx~\mbox{d}y.
\]
This again allows using a set of binary decision variables $\alpha_i$, to reformulate \eqref{eqT2} as
\begin{align}\notag
	\minimize_{\alpha_1,\ldots,\alpha_n}&~~\int_D\mbox{d}\vx - \int_{D\times [0,1]} \min\left( \sum_{i=1}^n \alpha_i \pi_{f_i}(\vx,y),1\right)\mbox{d}\vx~\mbox{d}y\\ \notag \mbox{subject to:}& ~~\sum_{i=1}^n\alpha_i \leq T\\ 
	& ~~\alpha_i\in\{0,1\}, ~~~i\in[n].\label{eqT3}
\end{align}
A major advantage of \eqref{eqT3} over the original program (\ref{eqT1}) is that the objective in \eqref{eqT3} is convex in $\alpha_i$, and a convexification of the combinatorial program is easily possible by relaxing the binary constraints to $0\leq\alpha_i\leq 1$. As a matter of fact, program (\ref{eqT3}) is a special case of the shape composition problem proposed and implemented in \cite{ aghasi2015convexa, aghasi2015convexb, redo2016terahertz,aghasi2018extracting}. Specifically, in \cite{aghasi2018extracting} the authors present a general set of sufficient conditions under which the convex relaxation accurately identifies the solution to the combinatorial problem. Furthermore, the convex relaxed program can be cast as a linear program (LP) and solved very efficiently. In the next section we will discuss the process of accurately solving \eqref{eqT3} by either solving an LP with binary constraints or an LP representing the relaxed program. 

\subsection{An LP Representation of the Problem}
Consider uniformly discretizing the domain $D\times[0,1]$ into $p$ voxels (or pixels in 2D): $(\vx_1,y_1),\ldots,(\vx_p,y_p)$. We can form a matrix $\bPi \in\mathbb{R}^{p\times n}$, where each column corresponds to a  discrete representation of  $\pi_{f_i}$. To reformulate \eqref{eqT3} as an LP, consider a variable $\bbeta\in\R^p$ with the entries 
\begin{align}\notag 
	\beta_j &= \min \left( \sum_{i=1}^n \alpha_i \pi_{f_i} (x_j,y_j),1\right)  \\ & = \min\left(\ve_j^\top \bPi \balpha ,1 \right),~~j\in[p],\label{eqT4}
\end{align}
where $\ve_j$ is the $j$-th standard basis vector. \eqref{eqT4} would naturally imply $\ve_j^\top \bPi \balpha\geq \beta_j$ and $\beta_j\leq 1$. Representing the integral with a sum and dropping the constant term $\int_D\mbox{d}\vx$), program (\ref{eqT3}) can now be cast as the following mixed binary program (MBP):
\begin{align}
	\minimize_{\balpha,\bbeta }&~ - \boldsymbol{1}^\top \bbeta ~~ ~ \mbox{subject to:}~ \left\{\begin{array}{l}\bbeta\leq \bPi\balpha ,  ~~~ \boldsymbol{1}^\top \balpha\leq T  \\ \bbeta\leq \boldsymbol{1},~~~~~~~ \alpha_i\in\{0,1\} ~~ i\in[n] \end{array} \right., \tag{MBP}\label{eqT5}
\end{align}
which offers the straightforward LP relaxation 
\begin{align}
	\minimize_{\balpha,\bbeta }&~~ - \boldsymbol{1}^\top \bbeta \quad  \mbox{subject to:}\quad \left\{\begin{array}{l}\bbeta\leq \bPi\balpha, ~~~ \boldsymbol{1}^\top \balpha\leq T  \\ \bbeta\leq \boldsymbol{1},~~~~ \boldsymbol{0}\leq \balpha\leq \boldsymbol{1} \end{array} \right.. \tag{LP}\label{eqT6}
\end{align}

We propose using Gurobi \cite{gurobi} to address \eqref{eqT5}. While Gurobi is capable of solving the combinatorial problem (\ref{eqT5}) using integer programming routines, our strategy is to start by solving \eqref{eqT6} first (which is computationally much faster than the original combinatorial problem). If the $\balpha$ component of the solution to \eqref{eqT6}  is binary, that solution automatically would correspond to \eqref{eqT5} as well, otherwise we try solving \eqref{eqT5} using the integer programming tools of Gurobi. 

A MATLAB implementation of our algorithm is available online. An overview of the implementation along with the link to the code is provided at the last section of the Supplementary Notes.  It is noteworthy that while our code handles both \eqref{eqT5} and \eqref{eqT6}, thanks to the tight relaxation, in all the experiments performed in this paper the LP relaxation produced binary solutions (which naturally correspond to the solutions of \eqref{eqT5}). Moreover, in the program overview, we have discussed a way to condense the matrix $\bPi$, without any change to the solution (see Section \ref{Supp:Condense} of the Supplementary Note \ref{Supp:secImplementation}). Even when $p$ (the number of voxels) is in order of millions, after condensing $\bPi$ our program can find a solution in fractions of a second, on a standard desktop computer. This is essential in applications such as real time foveated rendering in headmounted displays or lightfield displays. Meaning we can calculate which set of sparse depth to show to minimize the error per frame.   


\section{Assessment and Results}
We use the DoF profile patterns proposed in \cite{anderson2008minus}, which are representable in terms of the distance from the eye and the age. To generate a train of DoF profiles, we consider a uniform diopter spacing between the DoF planes, ranging between $D_{\min}$ and $D_{\max}$. The value of $D_{\max}$ is calculated based on \cite{anderson2008minus}, and we choose $D_{\min}$ to be 0.5 D (corresponding to the maximum visual range of 2 meters). This is a range for indoor visual activities such as monitor use. We have also performed similar experiments for $D_{\min}=0.09$ D (i.e., visual range expanded to 11.1 meters), which are mainly moved to the Supplementary Note \ref{Supp:secMoncular} to abide by the paper length guidelines. In the calculation of the DoF profiles, we set the pupil diameter to $p=$3 mm and 2 mm, where the former is set for an average 250 nits monitor at 66 cm working distance, and the latter is considered for high dynamic range (HDR) monitors. Finally, we consider different ways of weighting the profiles based on the age distribution of the users, and the average monitor distance setting.  

We use Figure \ref{figDoFs}, to present a schematic of the plain DoF profiles, and the different ways they are weighted based on the age and the office working distance. For all the plotted profiles in this figure we have used $p=3$ mm and $D_{\min} = 0.5$ D, which yield a DoF train with 151 profiles considering uniform 0.044 D increments in the diopter domain. The vertical axis is the normalized maximum image intensity profile plotted versus the distance captured from a delta function (bright spot) positioned at that distance. The horizontal axis on the left is the depth, and the horizontal axis on the right corresponds to the user's age. Since plotting all 151 profiles occupies the entire domain, in Figure \ref{figDoFs}(a) we have shown a sparse subset of the DoF profiles with only 14 profiles uniformly picked from the original train.

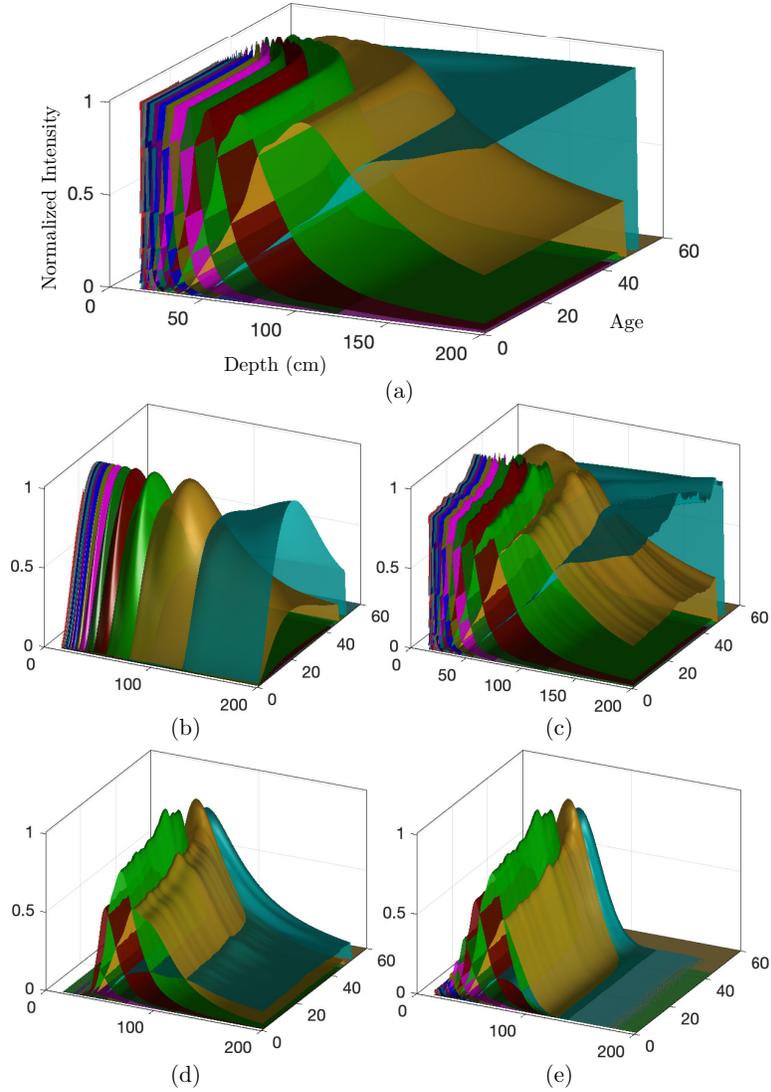
\begin{figure}[htb!]
	\centering \begin{overpic}[trim={.5cm .7cm  .4cm .1cm},clip,width=4in]{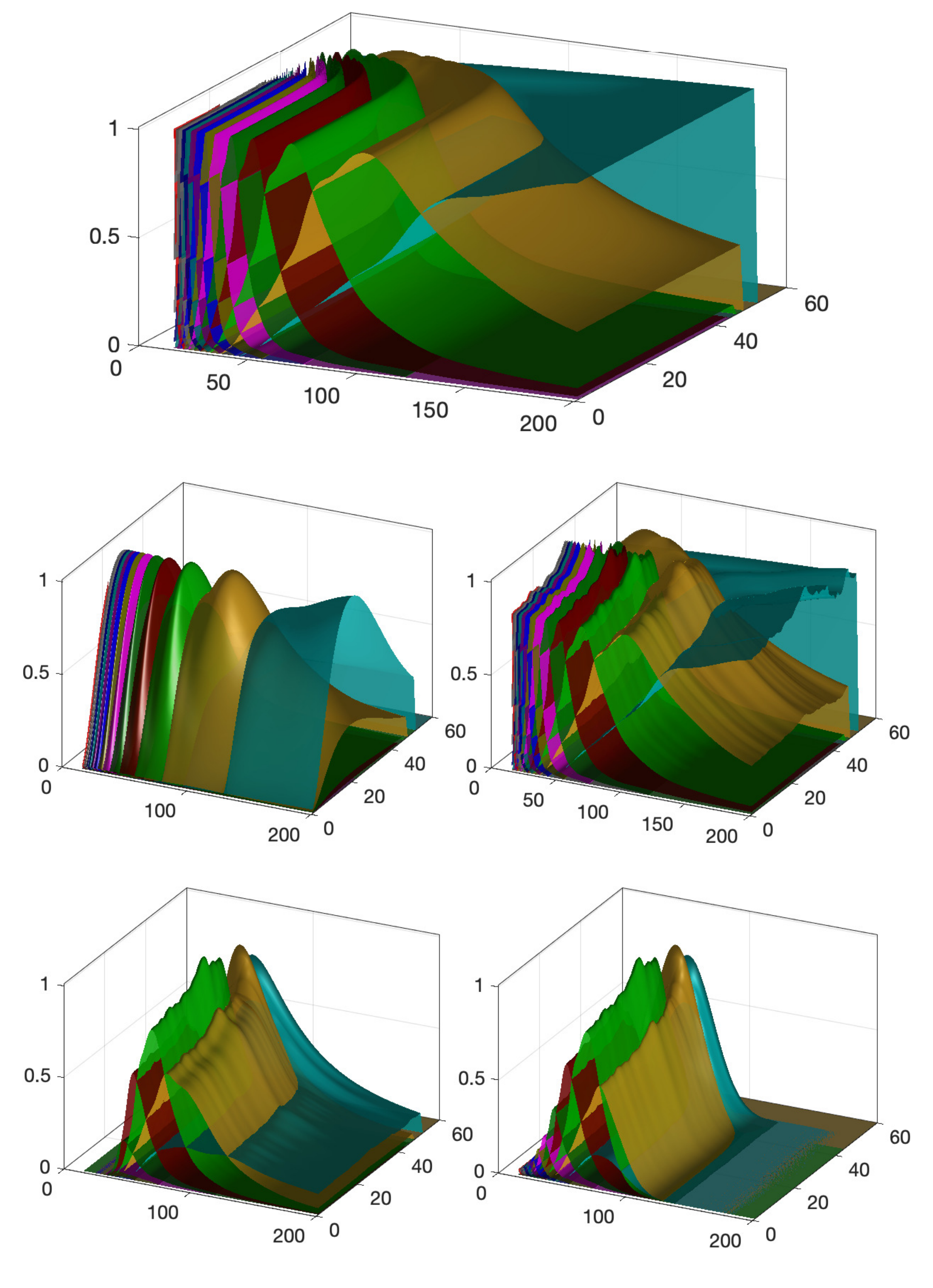}
		\put (35,62.5) {\scalebox{.8}{(a)}} 
		\put (15,31) {\scalebox{.8}{(b)}} 
		\put (50,31) {\scalebox{.8}{(c)}} 
		\put (15,-1.5) {\scalebox{.8}{(d)}} 
		\put (50,-1.5) {\scalebox{.8}{(e)}} 
		\put (3,72.5) {\rotatebox{90}{\scalebox{.7}{Normalized Intensity}}} 
		\put (20,65) {\rotatebox{0}{\scalebox{.7}{Depth (cm)}}} 
		\put (56,69) {\rotatebox{0}{\scalebox{.7}{Age}}} 
	\end{overpic}
	\vspace{.1cm}
	\caption{A schematic of 14 DoF profiles uniformly selected from a train of 151 profiles calculated for $p=3$ mm, with uniform diopter spacing between $D_{\min} = 0.5$ D and $D_{\max} = 7.08$ D. The right horizontal axis is age and the left horizontal axis is distance in mm. (a): The plain profiles without any weight on the depth or age. (b): The profiles after weighting the age axis by an empirical Gamma distribution. (c): The profiles after weighting the age axis by the US population. (d): Using the US population to weight the age, and a Gaussian weight in the diopter range. (e): The DoF profiles after weighting the age component by the US population and the depth range by a Gaussian.}\label{figDoFs}
\end{figure}

 Figure \ref{figDoFs}(b) shows the result of weighting the profiles by a desired target age distribution for which the display is being designed. Specifically, the users' age is considered to obey a Gamma distribution with a shape parameter $k=3$ and a scale parameter $\theta = 10$. This would correspond to a skewed bell-shaped distribution with mean 30 and a standard deviation 17.3, which models  the distribution that we empirically found for the users' age.  Figure \ref{figDoFs}(c) shows an alternative weighting of the profiles by age, which uses the US population, obtained from \cite{Census2019Age}.  Figure \ref{figDoFs}(d) shows a similar setting as  Figure \ref{figDoFs}(c), where aside from using the US population weight for the age, the diopter range is weighted by a Gaussian distribution with mean 1.5 D and standard deviation 0.5 D. Notice that such distribution introduces a reciprocal normal distribution weight on the depth axis. Finally,  Figure \ref{figDoFs}(e) shows a similar setting as  Figure \ref{figDoFs}(c), where the depth axis is weighted by a Gaussian with mean 66 cm and standard deviation 20 cm representing an average monitor distance setting.

We apply the optimization scheme proposed in Section \ref{section:math} to allocate the DoF planes for minimizing the accommodation error. Figure \ref{figOpt} shows the result of this optimization for different values of $T$, and for the parameter setting outlined in Figure \ref{figDoFs}. These results show, exactly, how the depth levels should be allocated for a band limited display starting from $T= 1$ monocular depth level, all the way to $T= 9$  levels, such that the accommodation error considering different factors such as age profile, working distance profile and brightness of the display is minimized. 

The optimized locations for the plain DoF train without a weight on the age or the depth are plotted with green circles as the background throughout all the sub-figures, to better help a comparison as different factors kick in. The red cross marks in Figure \ref{figOpt}(a) show the optimized locations when the age component of the profiles is weighted by a Gamma distribution as outlined above (a display product made for average age of 30 and standard deviation of 17.3). As noted, the target age is pushing the depth levels to shorter distances for smaller number of planes favoring younger eyes, but for larger number of planes the impact becomes more and more negligible.  

\begin{figure}[htb!]
	\centering \begin{overpic}[trim={2.5cm 2cm  2.cm 1.3cm},clip,width=4in]{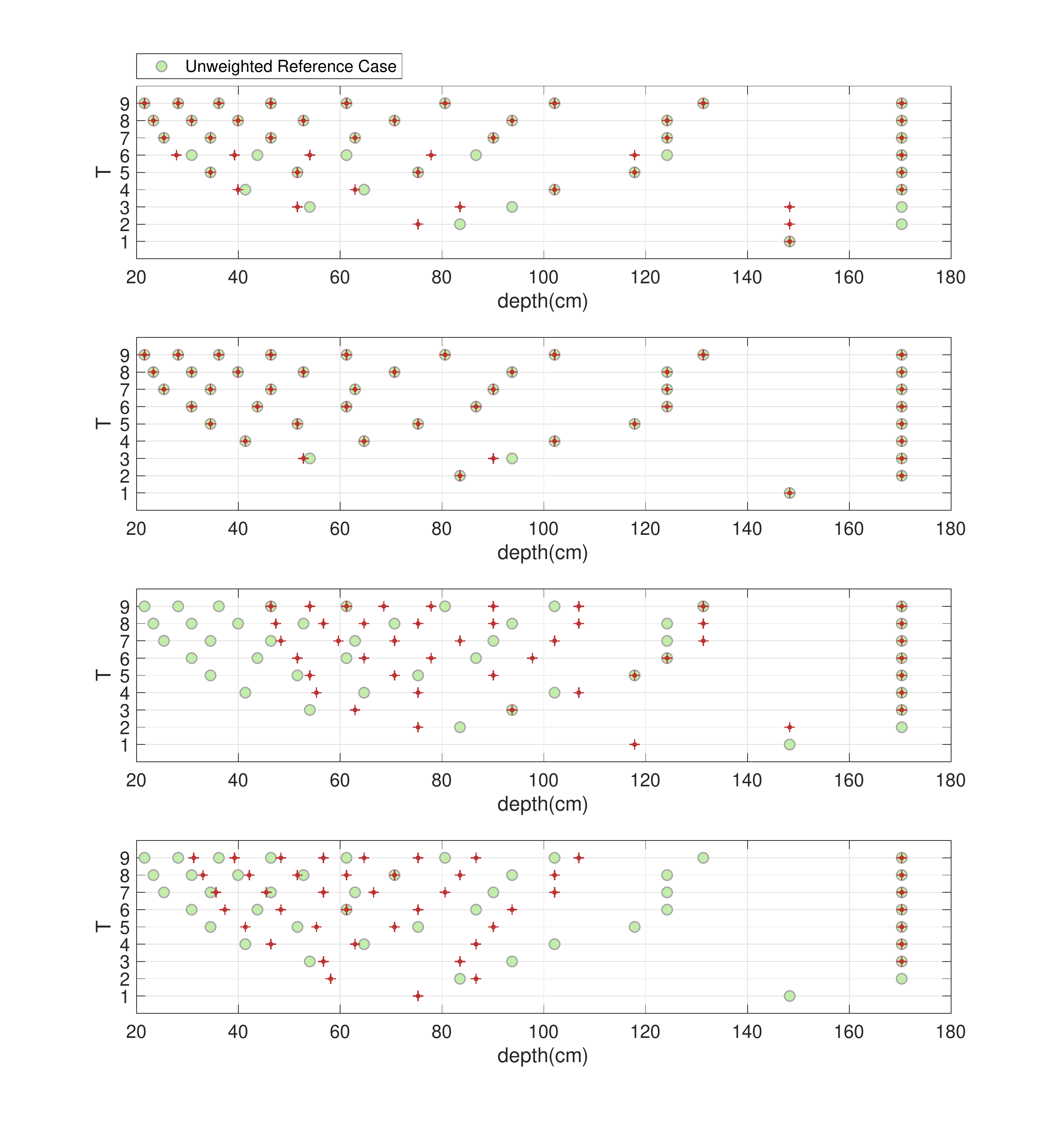}
		\put (34,73.5) {\scalebox{.8}{(a)}} 
		\put (34,48.5) {\scalebox{.8}{(b)}} 
		\put (34,24.5) {\scalebox{.8}{(c)}} 
		\put (34,-1.5) {\scalebox{.8}{(d)}} 
	\end{overpic}
	\vspace{.1cm}
	\caption{The optimized DoF plane allocation results for different values of $T$, and the parameter settings outlined in Figure \ref{figDoFs}. The optimized allocations for the plain DoFs without a weight on the age or the depth are plotted with green circles as the background in all the panels. The red crosses indicate the allocations after applying different age/depth weights. (a): The allocation after weighting the age axis by an empirical Gamma distribution. (b): The allocation after weighting the age axis by the US population. (c): Optimal allocation, using the US population to weight the age, and a Gaussian weight in the diopter range. (d): The allocation after weighting the age by the US population and the depth range by the Gaussian with center at 66 cm.
	}\label{figOpt}
\end{figure}
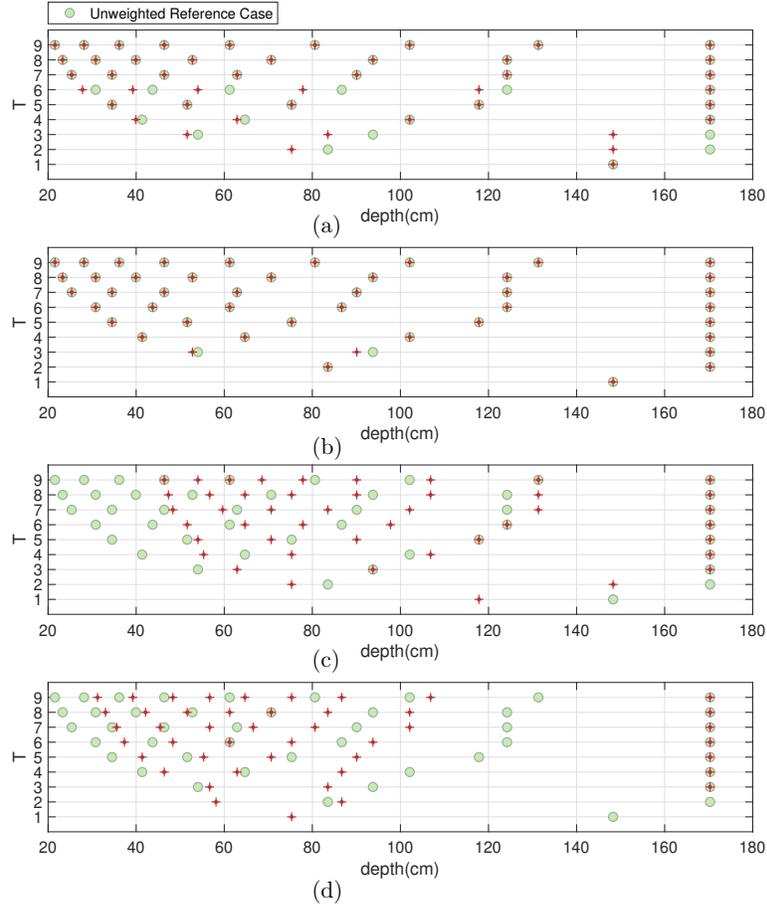

Figure \ref{figOpt}(b), corresponds to the weighting of the age by the US population impacting diopter ranges. The distribution is relatively flat, therefore, not much impact is noticed compared to a uniform age distribution. In a similar fashion,  Figure \ref{figOpt}(c) and (d) correspond to the settings described in panels (d) and (e) of Figure \ref{figDoFs}. As noted, we have a major shift of depth allocation in Figure \ref{figOpt}(c) because of the diopter range that we are targeting around 1.5 D and distance range that we are targeting around 66 cm. What is important here is the transition between the first three depth levels. It is evident that the population profile impact is less pronounced, if the display is designed for a certain working range. Another notable observation is how the optimization leans toward longer distances in lower number of depth levels. This figure  indicates that, for example, if there is a VR or an AR headset with two monocular depth levels designed for mass population, then a depth level around 170 cm, and the other around 83 cm, together introduce the minimal amount of accommodation error (panel(b)). However, if this headset is expected to represent a virtual monitor setup at 66 cm distance, then based on panel (d) the optimal depth choices are 58 cm and 87 cm. 

Figure \ref{figRefComp} shows a comparison of unweighted optimized levels for the case of a normal display (assumption of 3 mm pupil diameter) and an HDR enabled display with over 250 nits brightness (assuming a 2 mm pupil diameter). In reality the display may have a high dynamic range, so the pupil size varies based on the brightness of the content shown. Since for brighter contents the pupil diameter is smaller and the DoF is larger, therefore it is expected that less number of depth levels are needed to cover a desired diopter range. Here we are assuming that the content is so bright that the pupil diameter is fixed at 2 mm at all times. As seen, the depth levels become more spread apart, indicating a larger coverage per level. 

\begin{figure}[htb!]
	\centering \begin{overpic}[trim={.1cm .1cm  .1cm .1cm},clip,width=4.in]{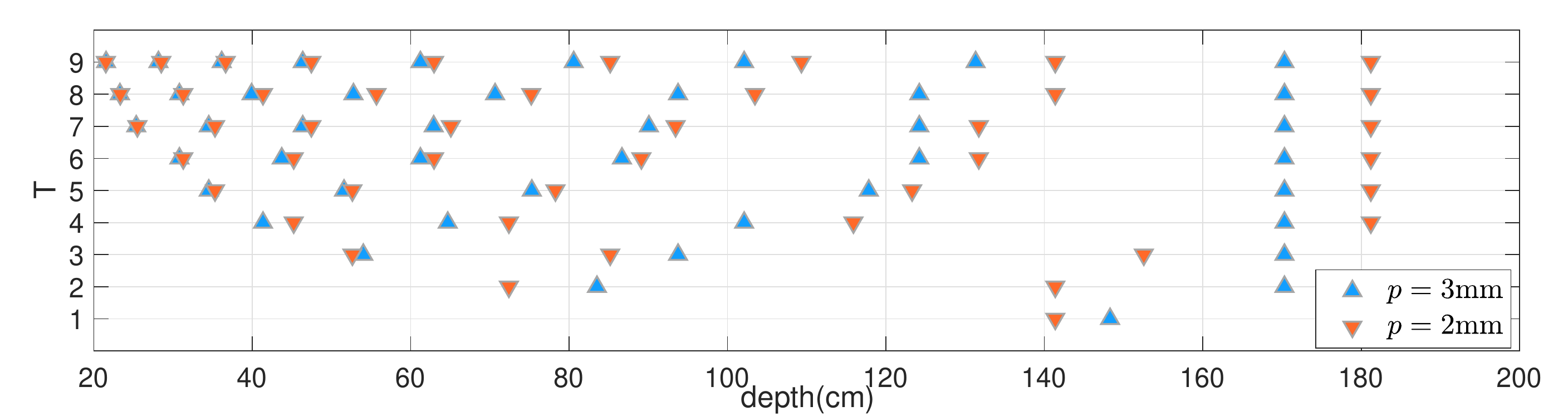}
	\end{overpic}
	\vspace{-.1cm}
	\caption{A comparison of the optimized DoF allocation for the unweighted reference cases, with $p=2$ mm and $p=3$ mm.}\label{figRefComp}
\end{figure}

Figure \ref{figOptHDR} a shows the impact of different factors on the allocation of the depth levels for 2 mm pupil diameter (HDR, or brighter displays). Figure \ref{figOptHDR}(a) shows the impact of age profiling (an average age of 30, and a standard deviation of 17.3). Figure \ref{figOptHDR}(b) shows the results for average US population, which is relatively uniform. As noted in both  panels (a) and (b), up to the three first levels, the age range almost has no impact in the optimized allocation of the depth levels. This is a significant observation, as one would think that because of dominant farsightedness in older population, one would need to shift the depth to further distances. For example, one would expect that allocating a depth level at 52 cm, when the bandwidth limits the depth levels to $T= 3$, would be undesired since almost all people after the age of 50 would lose acuity at that distance. But this optimization shows that such counter intuitive allocation still reduces the cost of accommodations error across the entire population. The population profile ultimately only slightly skews the optimized depth level allocations at higher $T$ values. Figures \ref{figOptHDR}(c) and (d) show the impact of distance and diopter targeting, respectively (similar depth and diopter profile considered as Figure \ref{figOpt}(c) and (d)). Because of the longer DoF at higher brightness, here again the impact of the profiling is less significant compared to pupil diameter of 3 mm. 

\begin{figure}[htb!]
	\centering \begin{overpic}[trim={2.5cm 2cm  2.cm 1.3cm},clip,width=4.5in]{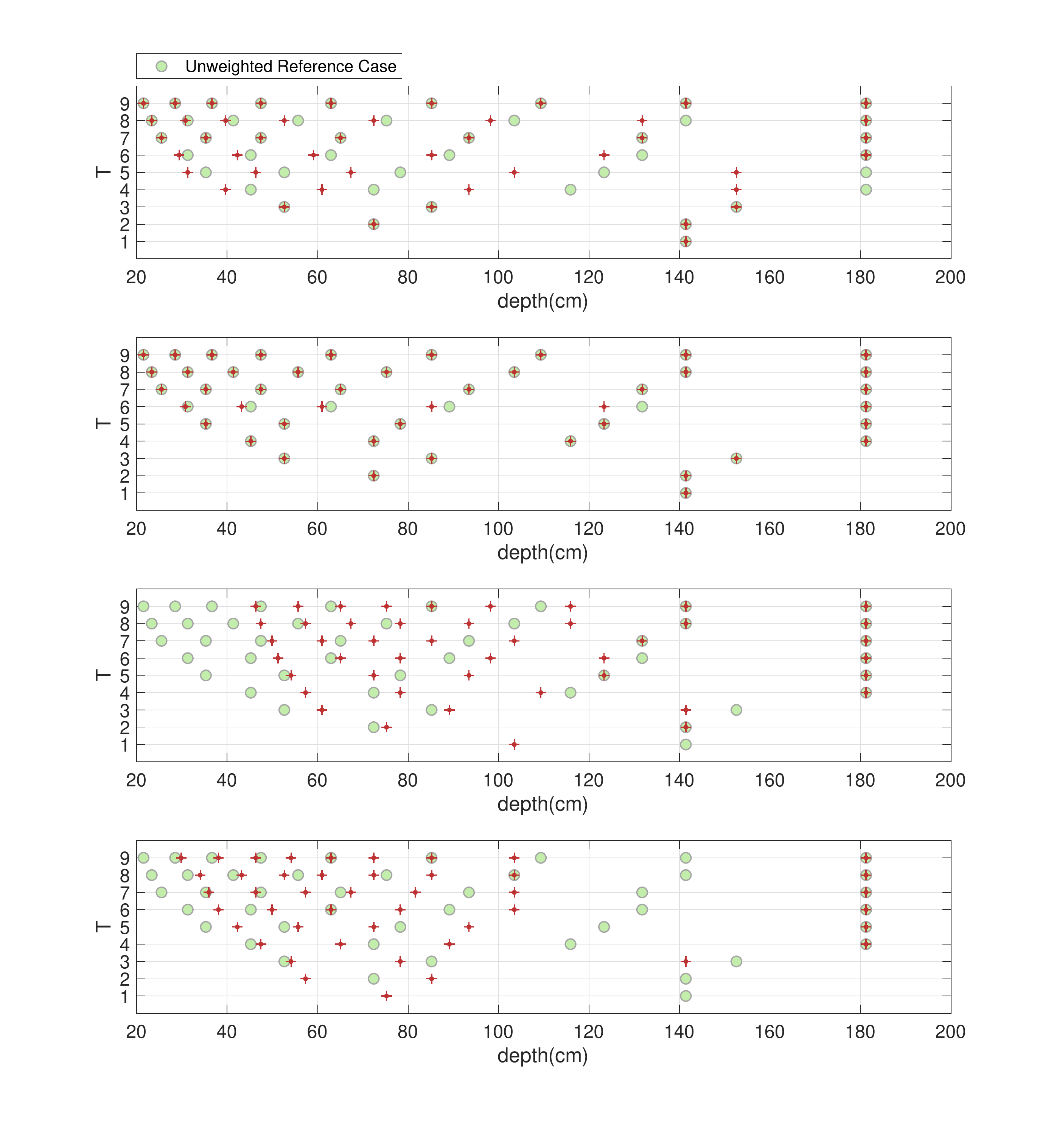}
		\put (34,73.5) {\scalebox{.8}{(a)}} 
		\put (34,48.5) {\scalebox{.8}{(b)}} 
		\put (34,24.5) {\scalebox{.8}{(c)}} 
		\put (34,-1.5) {\scalebox{.8}{(d)}} 
	\end{overpic}
	\vspace{.1cm}
	\caption{Optimized allocation results similar to the cases indicated in Figure \ref{figOpt}, with the pupil diameter set to $p=$2mm.}\label{figOptHDR}
\end{figure}
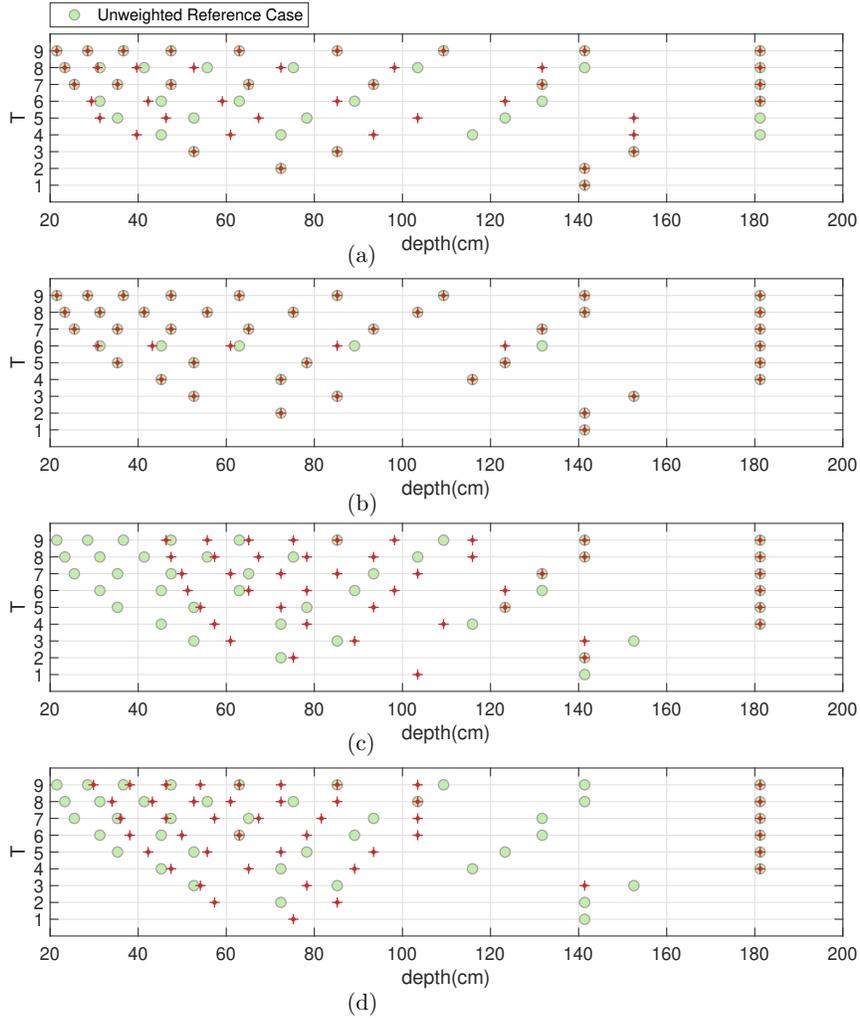

Similar analysis is done for a larger diopter (depth) range that goes all the way to 0.09 D (11.1 meters).  The results are shown in the Supplementary Figure \ref{Supp:figOpt}. In such a long range the optimization dominantly favors the last level in the range, in its selection to minimize the error. However, for $T=1$, the optimized allocation (not impacted by the age or application) is at 5.6 m (and 7 m for $p=2$ mm).   This is rather different compared to the 1.48 (1.41) meter distance when the optimization is ran only up to 2 meters. Certainly, if one expects the user to see mostly nearby 3D objects (as is the case of AR and VR), then the 2 meter range optimization is a better fit. If a display is desired to give the minimum accommodation error along the largest distance range all the way to infinity, then setting that first level to 5.5 meter is the best bet for an average 250 nits display. Supplementary Figure \ref{Supp:figGuy2} shows these depth levels compared to an average male body in one-to-one scale.

Figures \ref{figComp}(a-d) show the difference between the coverage errors for the different values of $p$ and $D_{\min}$ discussed in the paper. Coverage error is the accommodation error that is calculated between a continuum of accommodation (the real world), and what the eye perceives based on the quantization of accommodation, considering the DoF profile at each level. Mathematically, this error corresponds to the portion of the area (volume) of the box that is not covered by the union of the DoF profiles. Considering the iterative method (mentioned in Fig  \ref{fig1}), based on the shallowest measured DoF of 0.15 D at the largest pupil diameter of 6 mm, the human eye does not really discern a continuum of depth due to the limited acuity. Therefore, there is an intrinsic accommodation error in the image formed at the back of the eye at the highest acuity region of the fovea, which is not (at least based on psychophysics literature) perceivable to the human eye. One can consider this intrinsic error as the base error level of the human eye optics.  This intrinsic error is about 16.7\% of 2 meter range and about 1.8\% of the 11.1 meter range in distance space.

Here the blue bars show the coverage error for an approach where the depth levels are allocated by dividing the range to equal distances and red bars show the coverage error for optimized allocation. The horizontal axis is the number of depth levels, $T$. By comparing the results from optimized allocation of depth to the equidistant allocation, it is evident that the optimization is significantly reducing the coverage error. At the two meters distance range optimization (Panels (a) and (c) of Figure \ref{figComp}), the  error reduction may not seem as significant at the first glance, due to the existence of intrinsic errors. However, a closer look shows that for example with this optimized allocation for only $T=3$ depth levels, the coverage error is on par or better than all the 9 depth levels using equidistant allocation. This is a factor of three reduction in the needed number of layers based on this metric. The results become even more significant for longer range (11.1m) distance optimization, where clearly the coverage error is reduced up to factor of five for $T=9$ depth levels and only 2-3 level of optimized allocation can be on par with 9 equidistant allocations. This is a hard proof that uniform equidistant allocation of depth in diopter domain is a waste of bandwidth in lightfield displays with monocular depth, and such allocation becomes even more inefficient for larger number of layers. Also considering the intrinsic coverage error (16.7\% for panels (a), (c) and 1.8\% for panels (b), (d)) one can see that with only 6-7 number of optimized allocated depths, the coverage error is already starting to become unperceivable. What this means is that on average the lack of sharpness in the perceived image as a result of accommodation error becomes on par with the intrinsic lack of sharpness in the image due to DoF of the human eye at the given pupil diameter.

\begin{figure}[htb!]
	\centering \begin{overpic}[trim={.1cm .1cm  .1cm .1cm},clip,width=4.5in]{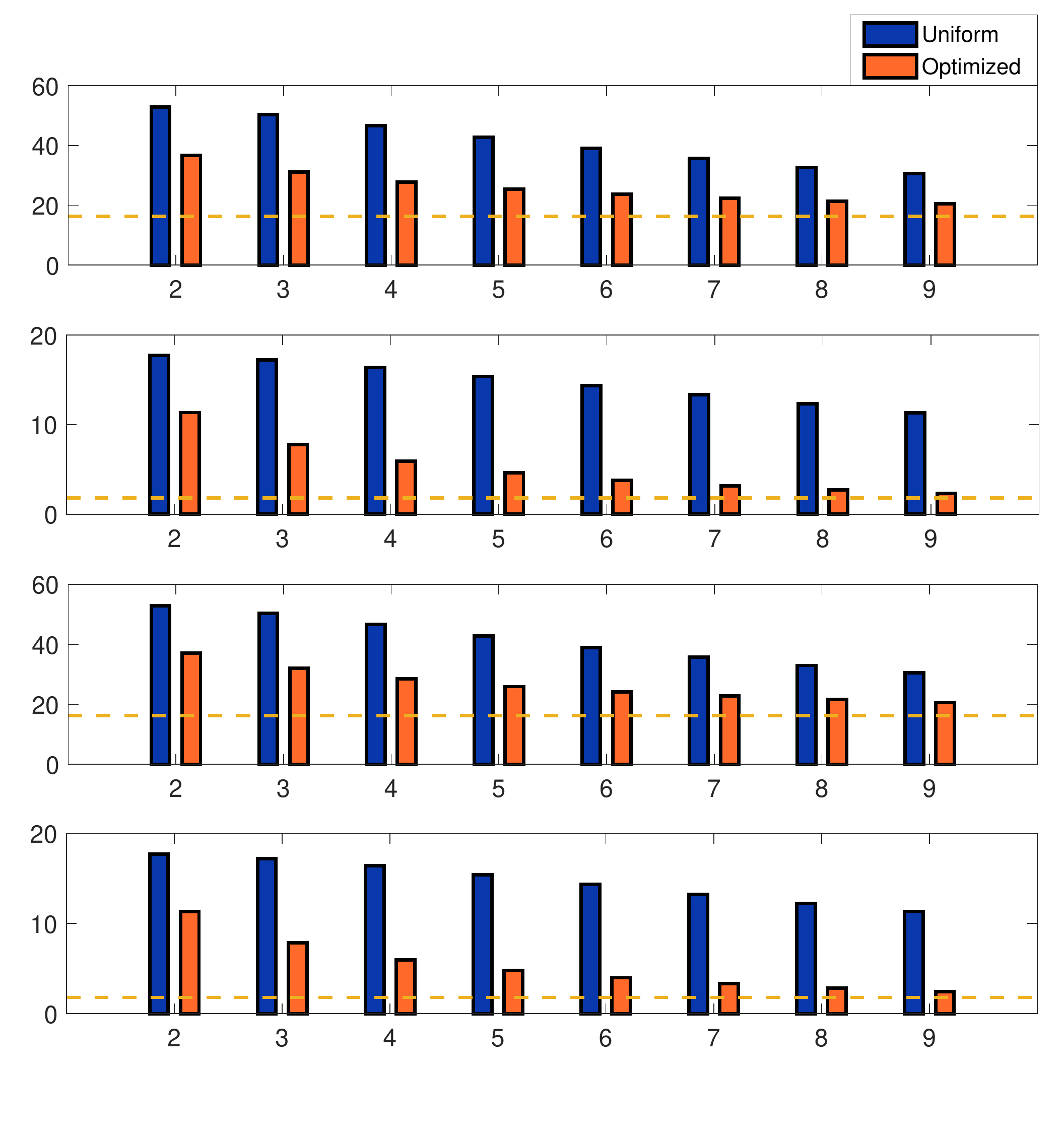}
		\put (48,72) {\scalebox{.8}{(a)}} 
		\put (84.5,82) {\scalebox{.7}{{\color{red} 16.7\%}}} 
		\put (48,50) {\scalebox{.8}{(b)}} 
		\put (85,57.2) {\scalebox{.7}{{\color{red} 1.8\%}}} 
		\put (48,28) {\scalebox{.8}{(c)}} 
		\put (84.5,38) {\scalebox{.7}{{\color{red} 16.7\%}}} 
		\put (48,1) {\scalebox{.8}{(d)}} 
		\put (85,13) {\scalebox{.7}{{\color{red} 1.8\%}}} 
		\put (45,5) {\scalebox{.7}{Depth Level ($T$)}} 
		\put (-1,45.5) {\rotatebox{90}{\scalebox{.7}{Coverage error (\%)}}} 
	\end{overpic}
	\vspace{-.4cm}
	\caption{Comparison between the coverage errors of equidistant DoF plane allocation in red, and the optimized allocation in blue (no weighting applied to the age or depth). The dashed red lines show the intrinsic errors along with their values. The setups are: (a) $p=$3 mm, $D_{\min} =$0.5D; (b) $p=$3 mm, $D_{\min} =$0.09D;
		(c) $p=$2 mm, $D_{\min} =$0.5D;
		(d) $p=$2 mm, $D_{\min} =$0.09D. }\label{figComp}
	
	
	
	
\end{figure}
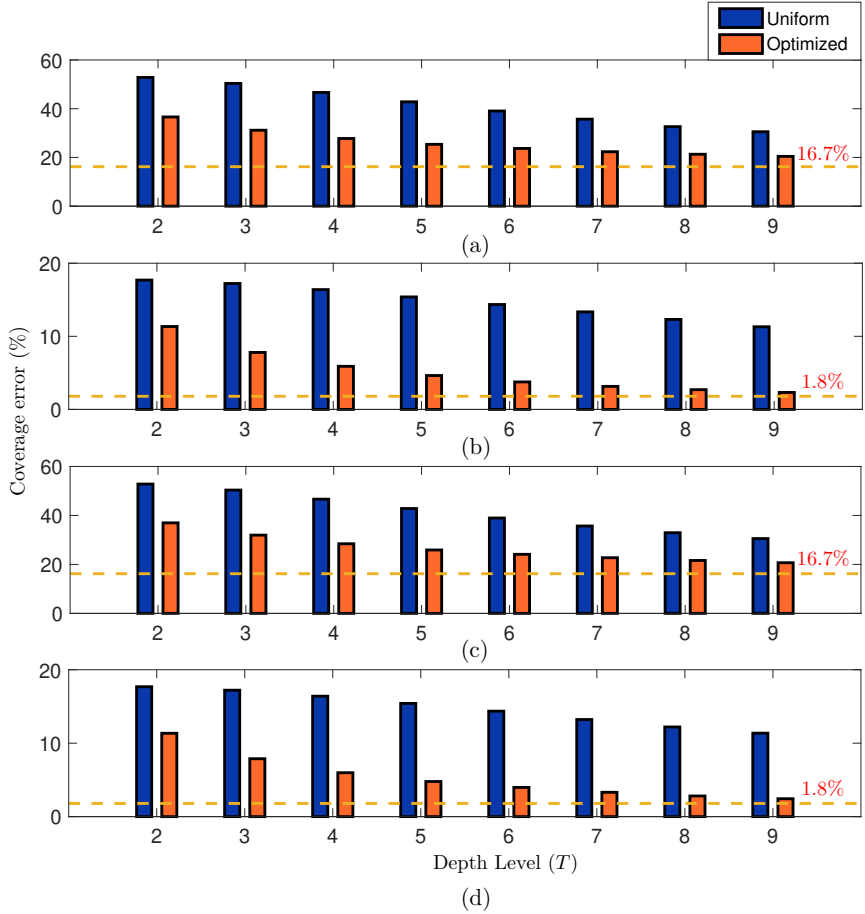

\section{Discussion}
\subsection{Quantized monocular depth levels in 3D space }
Another parameter of interest is the aberration profile of the eye for defining 3D shape of the monocular depth levels. This is especially more interesting for applications like foveated rendering in lightfield displays or head mounted displays, where the effort is to reduce the rendering computational cost by adapting to the eye nonuniform spatial acuity profile in 3D space. While there are thorough understanding of eye aberration and point spread function (PSF) profile \cite{wang2012statistical, carvalho2005accuracy, carvalho2002measuring}, it is difficult to theoretically pinpoint the single eye depth profile in the $(x,y)$ for each level. This is because the sampling outside the fovea becomes exponentially less dense, so defining a quantization parameter based on a single objective contrast sensitivity parameter becomes a multi-variable function of $(x,y,z)$ \cite{watson2014formula, navarro1998monochromatic, wong2015interpret}.

In case one considers the eye rotation (regardless of the type of eye movement), it can be roughly assumed that the monocular depth levels are on hemispherical surfaces with the eye rotation center at the center of the hemispheres. Obviously these spherical surfaces are then trimmed by the horizontal and vertical field of view of the eye into an irregular shape that varies based on structure of each person's nose. For fixated eye on this sphere, the monocular DoF gets larger and larger as the angle is increased from the fovea to parafovea to perifovea passing the macula boundary to near, mid, and far peripheral vision areas. Following a typical perimetry result for the eye visual field (the normal hill of vision) the monocular depth profile would be similar to a horn torus shape with concave sides hitting the minimum at the fovea region and a small hole at the place of the optic nerve head (Figure \ref{fig8} (a)) .  Certainly, there are many parameters such as the eye aberration; pupil diameter and retina curvature that will impact this shape, but that level of accuracy is most likely not needed to be considered for designing a band limited display and thus is not in the focus of our study. The mathematical formulation presented in Section \ref{section:math} is general for any number of dimensions. Hence, based on the 3D monocular depth profile, one can increase the number of variables in our optimization approach to find the optimal allocations in 3D space and find the coverage error in higher dimensions. 

\begin{figure}[!ht]\vspace{-20pt}
\centering
	\begin{overpic}[trim={4.3cm 3cm  2.cm 2cm},clip,width=2.18in]{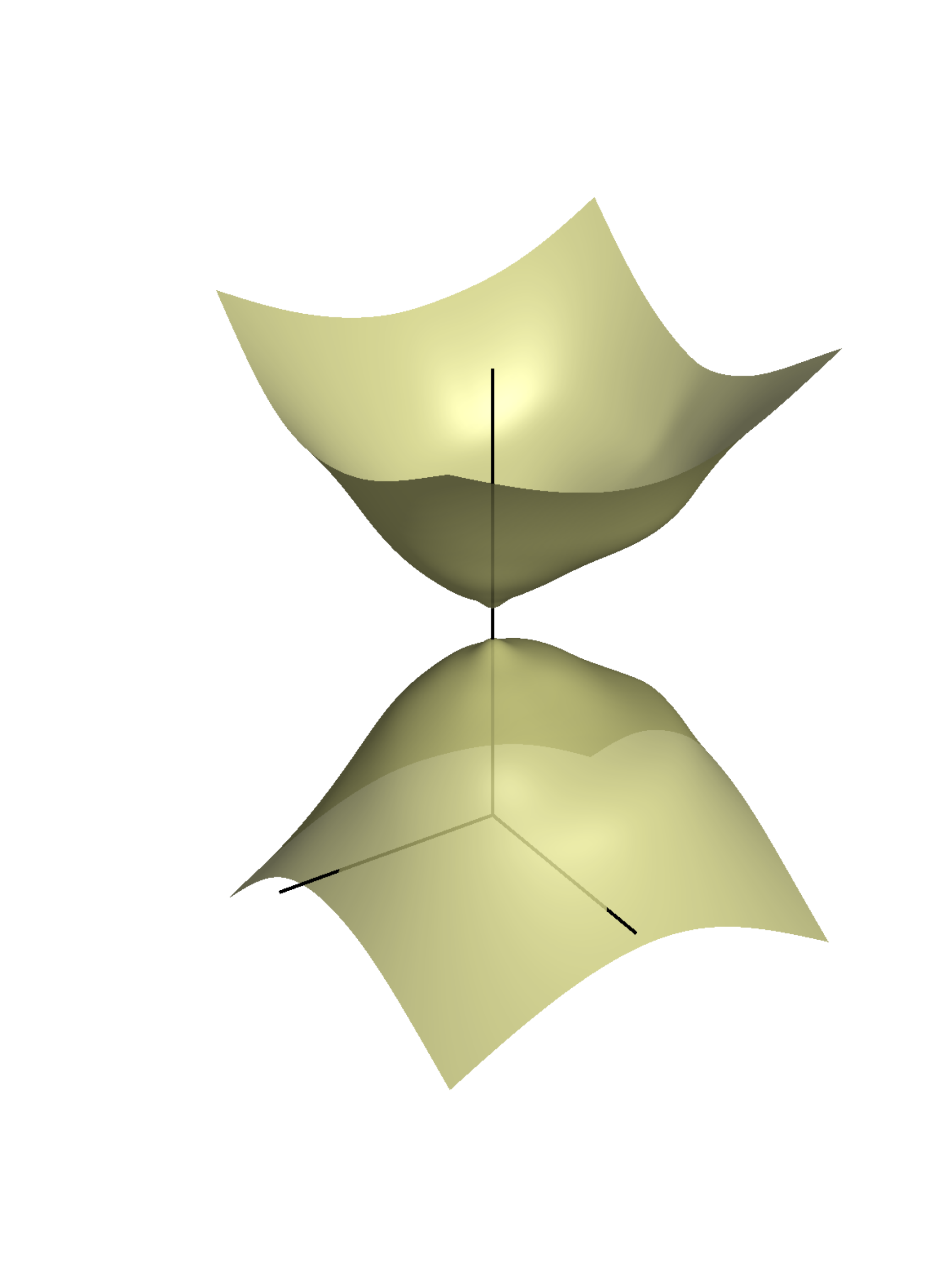}
	\end{overpic}
	\begin{overpic}[trim={.1cm .1cm  .1cm .1cm},clip,width=2.38in]{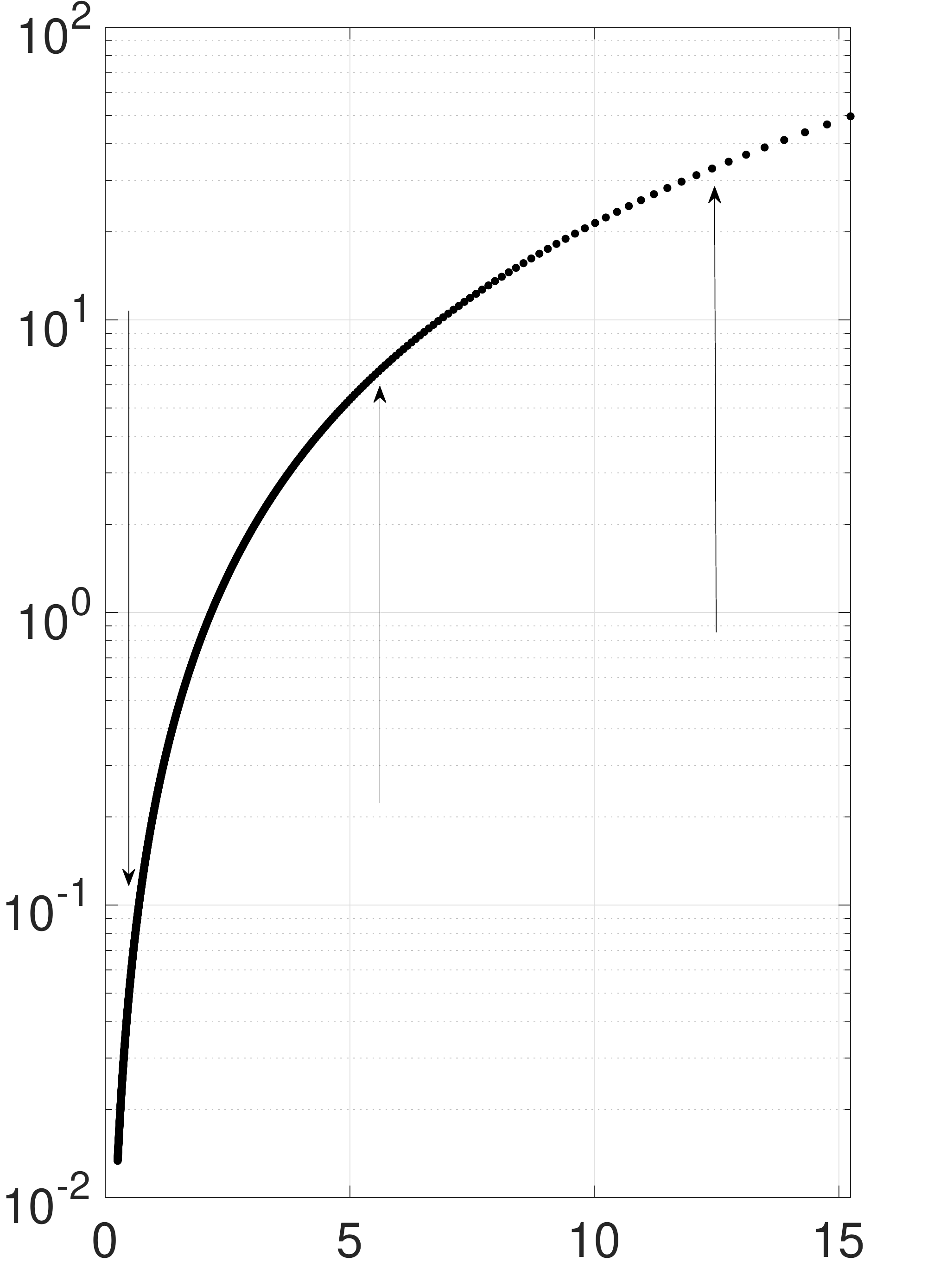}
	\end{overpic}
	\\[.45cm]
	\hspace{-.1cm}\begin{overpic}[trim={1cm 1.2cm  .5cm 1.2cm},clip,width=4.5in]{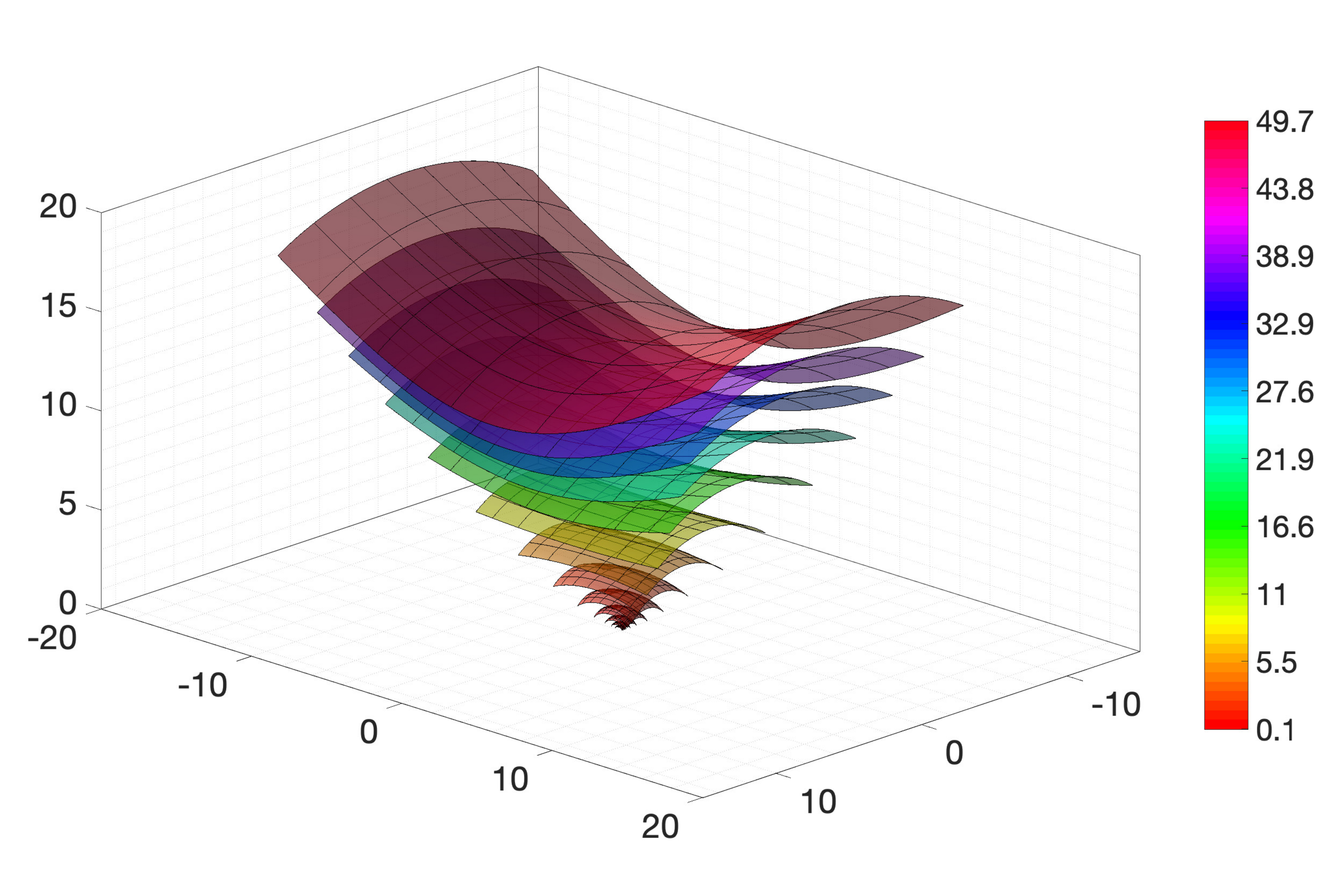}
		\put (23,63) {\scalebox{.8}{(a)}} 
		\put (73,63) {\scalebox{.8}{(b)}} 
		\put (48,-3) {\scalebox{.8}{(c)}} 
		\put (-1,83) {\scalebox{.8}{$x$}} 
		\put (33,81) {\scalebox{.8}{$y$}} 
		\put (15,125) {\scalebox{.8}{$z$}} 
		\put (60,66) {\scalebox{.7}{Depth of horopter (m)}} 
		\put (47,93.5) {\rotatebox{90}{\scalebox{.7}{Discernible $\Delta z$ (cm)}}} 
		\put (9,7) {\scalebox{.7}{$x$ coordinate}}
		\put (6,4) {\scalebox{.7}{at the frontal plane (m)}} 
		\put (-3,24.5) {\rotatebox{90}{\scalebox{.7}{Depth of horopter (m)}}} 
		\put (57,120.5) {\rotatebox{90}{\scalebox{.7}{PSF dominates}}}
		\put (60,123.5) {\rotatebox{90}{\scalebox{.7}{the acuity}} }
		\put (71,83.5) {\rotatebox{90}{\scalebox{.7}{mid range depth}} }
		\put (89.5,80.5) {\rotatebox{90}{\scalebox{.7}{long depth range dominated }}}
		\put (92.5,79.5) {\rotatebox{90}{\scalebox{.7}{by negative super linear error}}} 
		
	\end{overpic}
	\caption{(a): A schematic of the Traquair hill of vision. (b): Stereoscopic distinguishable depth variation in longitudinal horopter vs depth of observation distance.  (c): The 3D horopters in space based on empirical models from Ogle and Helmholtz measurements \cite{ogle1932analytical, ogle1956stereoscopic, ames1932corresponding} . Color indicates the depth change in centimeters that is distinguishable on that horopter. 
	}
	\label{fig8}
\end{figure}

\subsection{Generalization to binocular depth resolution}

Similar to the monocular depth in Figure \ref{fig1}, using horizontal angular disparity $\delta$ and its relation with interpupillary distance (IPD) and depth, one can iteratively find the binocular depth levels (horopters) and use our optimization approach to allocate limited number of horopters for a band limited stereoscopic display. This can be used to compress the stereoscopic data not based on the data redundancy, but based on the human eye vision limits. The maximum stereoscopic acuity is reported from 0.167 arcmin max to 0.5 arcmin average in the literature \cite{larson1983stereoscopic,palmisano2010stereoscopic}. However this acuity is variant with angle with regards to the center line that is perpendicular to the face plane. This is why the horopter levels are different at different angles.  Assuming the experimental stereoscopic acuity as 0.5 arcmin, and the average pupil distance \cite{larson1983stereoscopic,palmisano2010stereoscopic}, $I$, equals 64 mm, at each observation distance, $z$, we calculate the physical depth between two closest objects that are detected to be at different depths (see Supplementary Note \ref{Supp:Helmoltz}).  We find the next observation distance or horopter level by substituting the distance with $z+\Delta z$. The iteration starts from 25 cm and terminates when the observation distance is over 15 m. 

As shown in Figure \ref{fig8} (b), human eyes are very sensitive at differentiating the depth levels on the order of sub-millimeter when the observation distance is near (less than 66 cm). This is only for the center line of the stereoscopic depth levels or horopters at maximum acuity. Based on this iterative simulation, one can find the center point of the stereoscopic depth levels, their relative location and their total number. The total number of levels from 25 cm to 15m is 1731 (for this IPD = 64 mm) and based on ANSUR IPD data \cite{dodgson2004variation} related to different races this number varies with 7.7\% between races. There is minimal effect from monocular DoF on the number of these levels in the 25 cm-15 m range \cite{larson1983stereoscopic}. If one ignores the effect of eye PSF \cite{ginis2012wide} down to 15 cm to the face, then the number of levels is increased to 2905 from 15 cm to infinity, but as pointed in the literature the shallowest depth that is detectable based on human vernier acuity (hyper acuity) cannot be smaller than 100 micron \cite{lit1959depth, westheimer1990contrast}. This will cap the depth levels to 2667. Above the 15 m distance, stereoscopic depth perception shows super linear accumulative negative error \cite{palmisano2010stereoscopic}. For example at 15 m this error is about -5 m so the subjects detect the 15 m to be closer around 10 m and at 31 m they perceived the depth to be around 18 m and at 248 they perceived the depth to be only 50 m. These data indicate that there are not many more depth levels that can be quantized after 15 m all the way to infinity, at least based on the non-contextual monocular or stereoscopic cues. Figure \ref{fig8} (b) further informs these three different regimes of stereoscopic depth perception (near distances dominated by PSF, mid range that is iteratively quantized, far that has negative accumulative error).  Supplementary Figure \ref{Supp:figGuy3}(a) shows these depth levels at one-to-one scale with regards to an average 40 years old, male anatomy. 

\subsection{Quantized binocular depth levels in 3D space }
In order to find the 3D spatial profile of binocular depth levels we used the Ogle \cite{ogle1932analytical, ogle1956stereoscopic} equation for horopters along with the Helmholtz experimental results as in Supplementary \eqref{HelmEq}  \cite{ames1932corresponding}. This estimation, unlike the Vieth-Muller circle theory, considers the Hering-Hillebrand deviation. To find all $H$ values at an arbitrary distance, we performed a linear fit with respect to the dioptric distance. The resulted shape of horopters are shown in  Figure \ref{fig8}(c) (see the supplementary video for full 3D renders of the horopters and the Supplementary Figure \ref{Supp:figGuy3}(b)). Horopters get denser as the distance becomes closer to the face. At shorter observation distances, horopters take the convex shape and gradually become concave on x axis while still keeping a convex shape on y.  Based on this result it i
s evident that a flat-screen or virtual image is significantly sub-optimal.  A more accurate empirical measurement is recently done in \cite{gibaldi2019binocular,gibaldi2017active} to study the different mechanisms of saccade and vergence from a neurological view point.

Finally, in order to experimentally validate the iterative depth quantization criteria in this study a Badal setup as in Supplementary Figure \ref{Supp:figExp} can be used to measure the number of monocular depth levels for subjects. The experiment has to run on different age groups and at different lighting conditions to testify the optimized allocation of the levels. This is a topic of our future studies. One can expand the same approach for allocation of binocular horopters.

\section{Conclusion}
We defined a method to quantize and allocate monocular depth levels in an optimized fashion. The method sets a fundamental guideline for designing 3D displays based on human visual perception capacity. From display science perspective, this is essentially equivalent to human eye depth resolution. The iterative depth allocation approach results in maximum 40 monocular depth levels and 2667 binocular levels which saturate even the best and youngest eye at most demanding experimental conditions possible. This maximum quantized discernible levels quickly falls to 15 monocular levels and 1731 binocular spread from 25 cm to 15m for 3 mm pupil diameter and IPD of 64 mm. To optimally allocate monocular depth levels to a lightfield or a 3D display, we manage to cast the problem to an integer program with a computationally accurate and efficient convex relaxation. In some cases, the results beat an equidistant allocation of monocular depth  levels in diopter space by a factor of 5. Variety of design parameters are studied through this study. Our method shows that with only 6-7 quantized monocular depth levels allocated optimally for an AR/VR headset or a lightfield display application one can have an on-par error with intrinsic error of human eye. The method goes on further to pin point the location of these depth levels for varieties of desired scenarios.

\section*{Acknowledgements}
The authors acknowledge Adam Samaniego, Dr. Youngeun Park, and Dr. Stefano Baldasi from the analytic department at former Meta Co. for their consultation.

\section{Supplementary Notes}
\subsection{Supplementary Note 1: Monocular depth allocation for longer range coverage }\label{Supp:secMoncular}
This section presents some more allocation results for longer range, where the minimum diopter range is set to $D_{\min}=0.09$ D (i.e., visual range expanded to 11.1 meters). As before, the value of $D_{\max}$ is calculated based on \cite{anderson2008minus}. Similarly, two different scenarios in terms of the pupil diameter are considered: $p=$3 mm and 2 mm, where the former is set for a standard monitor, and the latter is considered for a high dynamic range (HDR) monitor.

The optimization scheme proposed in Section \ref{section:math} is applied to allocate the DoF planes for minimizing the accommodation error. The left and right columns in Figure \ref{Supp:figOpt} show the result of this optimization for different values of $T$. The left column corresponds to $p=3$ mm, whereas on the right column $p=2$ mm is considered. For both columns, the first row  considers plain profiles without any weight on the depth or age, the second row considers weighting the age by a Gamma distribution, and the third and fourth rows correspond to the weighting of the age by the US population, and a Gaussian weighting in the diopter range and depth range, respectively. The parameters of the Gamma and Gaussian distributions used for weighting are identical to those used for Figure \ref{figOpt}.

\begin{figure}[htb!]
	\centering \begin{overpic}[trim={1.5cm 2cm  1.3cm 1.3cm},clip,width=2.63in]{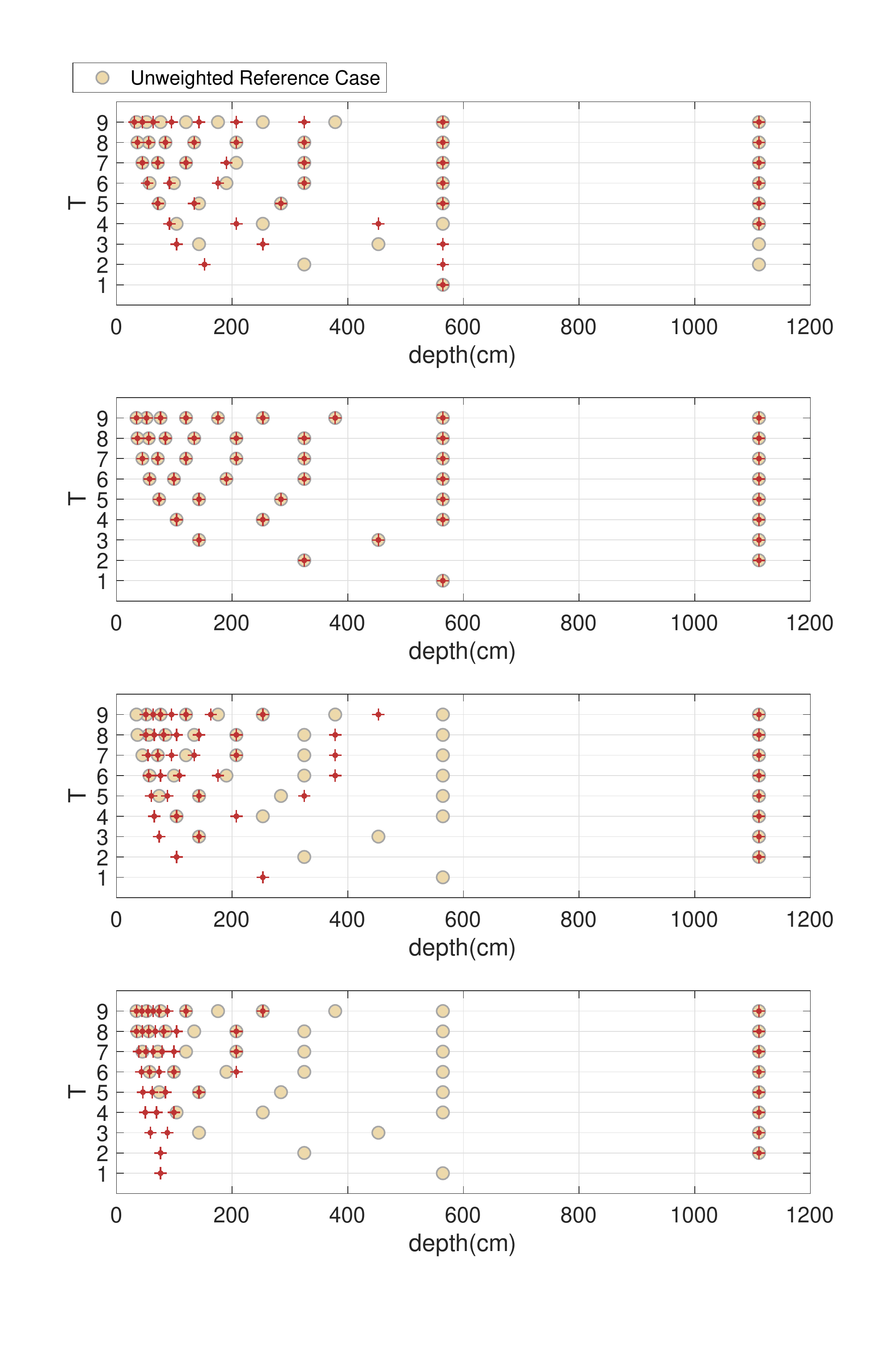}
		\put (25,73.5) {\scalebox{.8}{(a)}} 
		\put (25,48.5) {\scalebox{.8}{(b)}} 
		\put (25,24.5) {\scalebox{.8}{(c)}} 
		\put (25,-1.5) {\scalebox{.8}{(d)}} 
	\end{overpic}\hspace{-.1cm}
	\centering \begin{overpic}[trim={1.5cm 2cm  1.3cm 1.3cm},clip,width=2.63in]{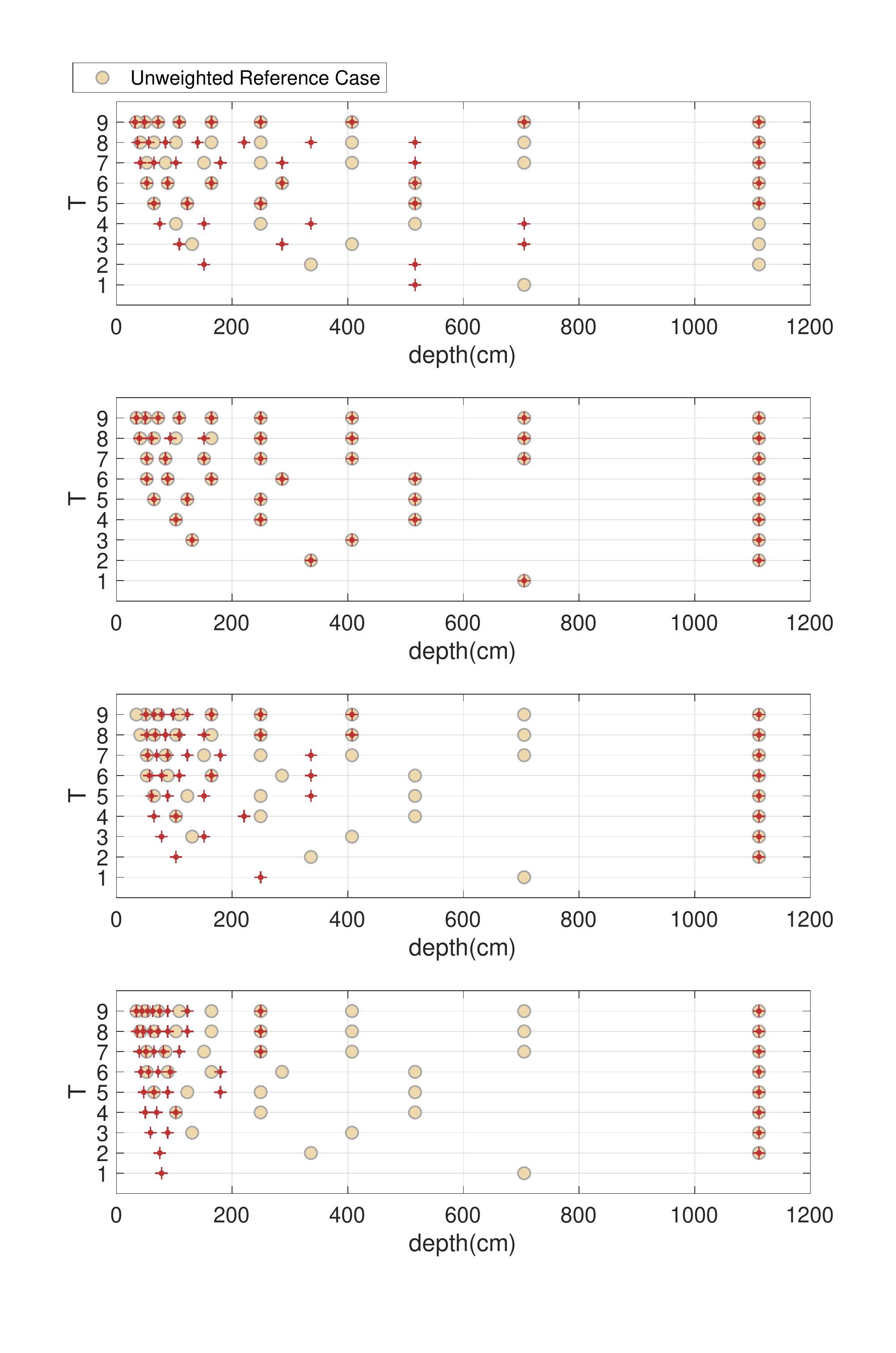}
		\put (38,73.5) {\scalebox{.8}{(e)}} 
		\put (38,48.5) {\scalebox{.8}{(f)}} 
		\put (38,24.5) {\scalebox{.8}{(g)}} 
		\put (38,-1.5) {\scalebox{.8}{(h)}} 
	\end{overpic}
	\vspace{.1cm}
	\caption{The left and right columns show similar stories for each row as detailed in Figure \ref{figOpt}. On the left panels the pupil diameter and the minimum diopter range are set to $p=3$ mm and $D_{\min} = 0.09$ D, and for the right column $p=2$ mm and $D_{\min} = 0.09$ D. }\label{Supp:figOpt}
\end{figure}
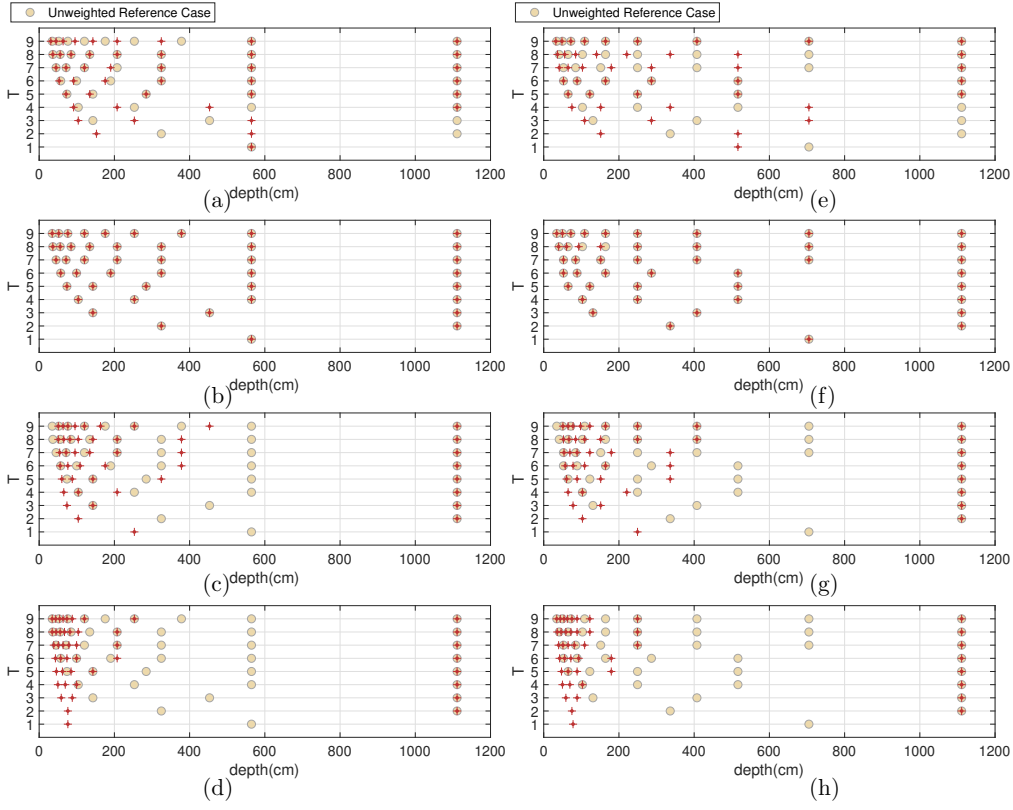

\subsection{Supplementary Note 2: Binocular horopters calculations}\label{Supp:Helmoltz}
\subsubsection{Iterative Approach}
The horizontal angular disparity $\delta$ can be approximated using
\begin{equation}
\delta = \frac{\Delta zI}{z^2 + z\Delta z},
\end{equation}
where $\Delta z$ is the depth step that is creating the angular disparity between the eyes, $I$ is the IPD, and $z$ is the depth or distance from the face (the line that connects the center of two pupils) to the depth plane. At each observation distance, $z$,  the physical depth between two closest objects that are detected to be at different depths are calculated. The next observation distance or horopter level is found by substituting the distance with $z+\Delta z$.

\subsubsection{Horopter 3D profile} 
The following equation shows Ogle formula for calculating the Hillebrand horizontal deviation:
\begin{equation}\label{HelmEq}
x^2\left( 1- H\frac{z}{2a}\right) + y^2\left( 1+H\frac{z}{2a}\right) - y\left( \frac{z^2-a^2}{b}+Ha\right) - a^2 + H\frac{az}{2} = 0.~~~~
\end{equation}
Here, $2a$ is the pupil distance, $z$ is the point of fixation of eyes to the center of two eyes, and $(x,y)$ is any point on the trace of the longitudinal horopter. Here, Hering-Hillebrand deviation $H$ is a constant at a fixation distance $z$. The degree of deviation from the theoretical to the empirical horopter trace depends on $H$. Table. \ref{tab:shape-functions} summarizes the variation of $H$ and its relation with distance and eye diopter. The vertical profile of the horopters is assumed to follow a spherical surface with center at the middle point of IPD.

\begin{table}[htbp]
	\centering
	\caption{Value of Hering-Hillebrand deviation for each fixation distance [47].}
	\begin{tabular}{c|c|c|c|c|c}
		\hline
		Fixation &4.5 m &7.3 m &2.37 m &1.29 m &5.5 m \\
		Distance (z) & 2.22 D &1.37 D &4.22 D &7.75 D &1.82 D\\
		\hline
		$H$ &0.108 &0.068 &0.203 &0.366 &0.086\\
	\end{tabular}
	\label{tab:shape-functions}
\end{table}

\subsection{Supplementary Note 3: Depiction of Monocular and Binocular Depth Level Distances}\label{Supp:MonoBin}
For a comparison of the depth levels in one-to-one scale, in this section we present some additional figures showing the monocular depth levels (Figure \ref{Supp:figGuy2}), and the binocular depth levels (Figure \ref{Supp:figGuy3}) relative to a male body.

\begin{figure}[htb!]
	\centering \begin{overpic}[trim={0cm 0cm  0cm 0cm},clip,width=4.5in]{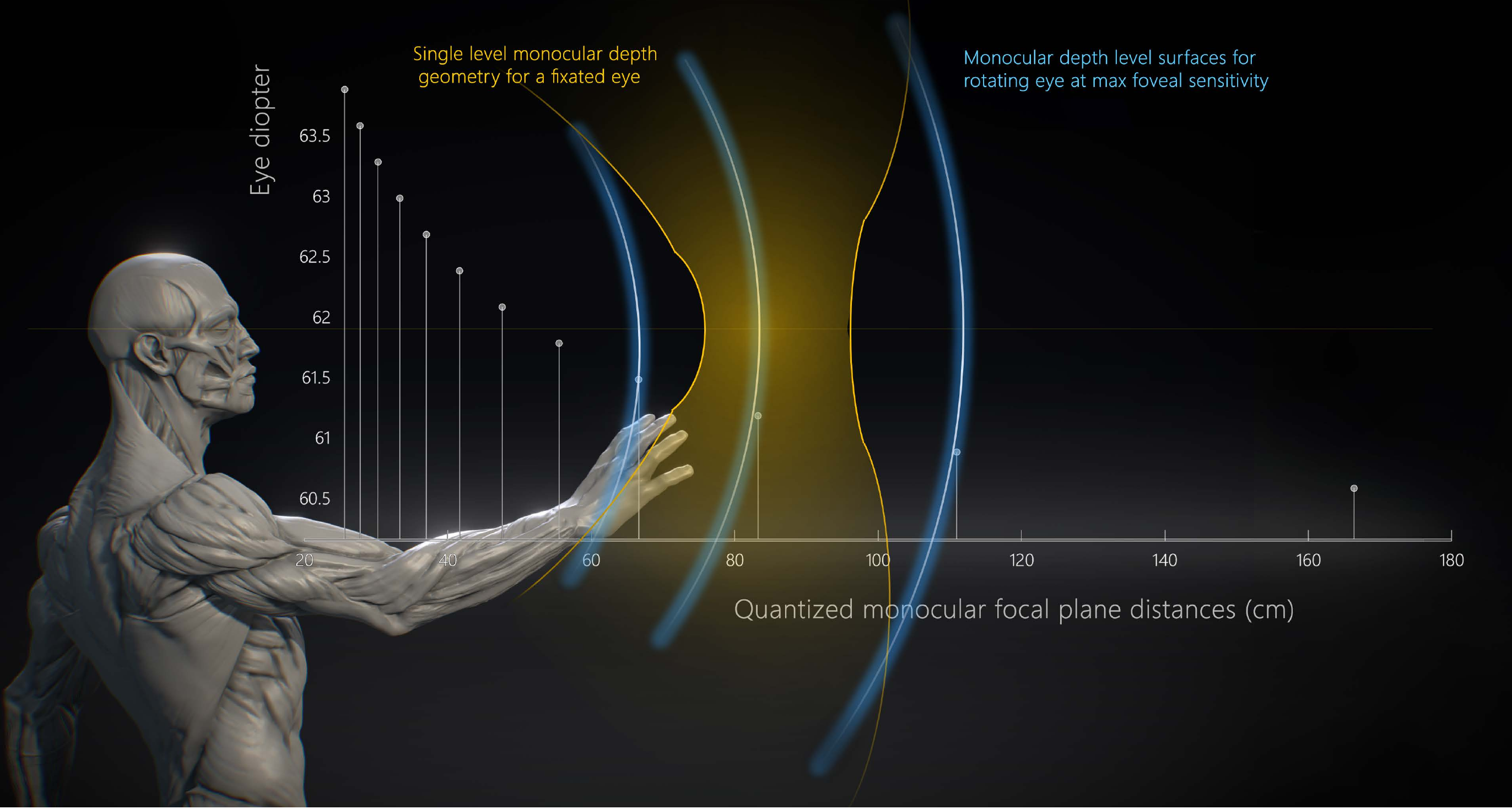}
	\end{overpic}
	\caption{Monocular depth level distances and schematic of the spatial profile are shown from the face of an average 40-year old eye with 3 mm pupil diameter in one-to-one scale with regard to the male anatomy. The blue curves show the surface of the hemisphere for each depth level tiled by the eye rotation at the center of each depth of field range that is perceived as one monocular depth. The orange area shows the depth of field spatial profile for a fixated eye. The area becomes larger as spatial contrast sensitivity drops at larger angles from the fovea.   }\label{Supp:figGuy2}
\end{figure}

\begin{figure}[htb!]
	\centering \begin{overpic}[trim={0cm 0cm  0cm 0cm},clip,width=4.5in]{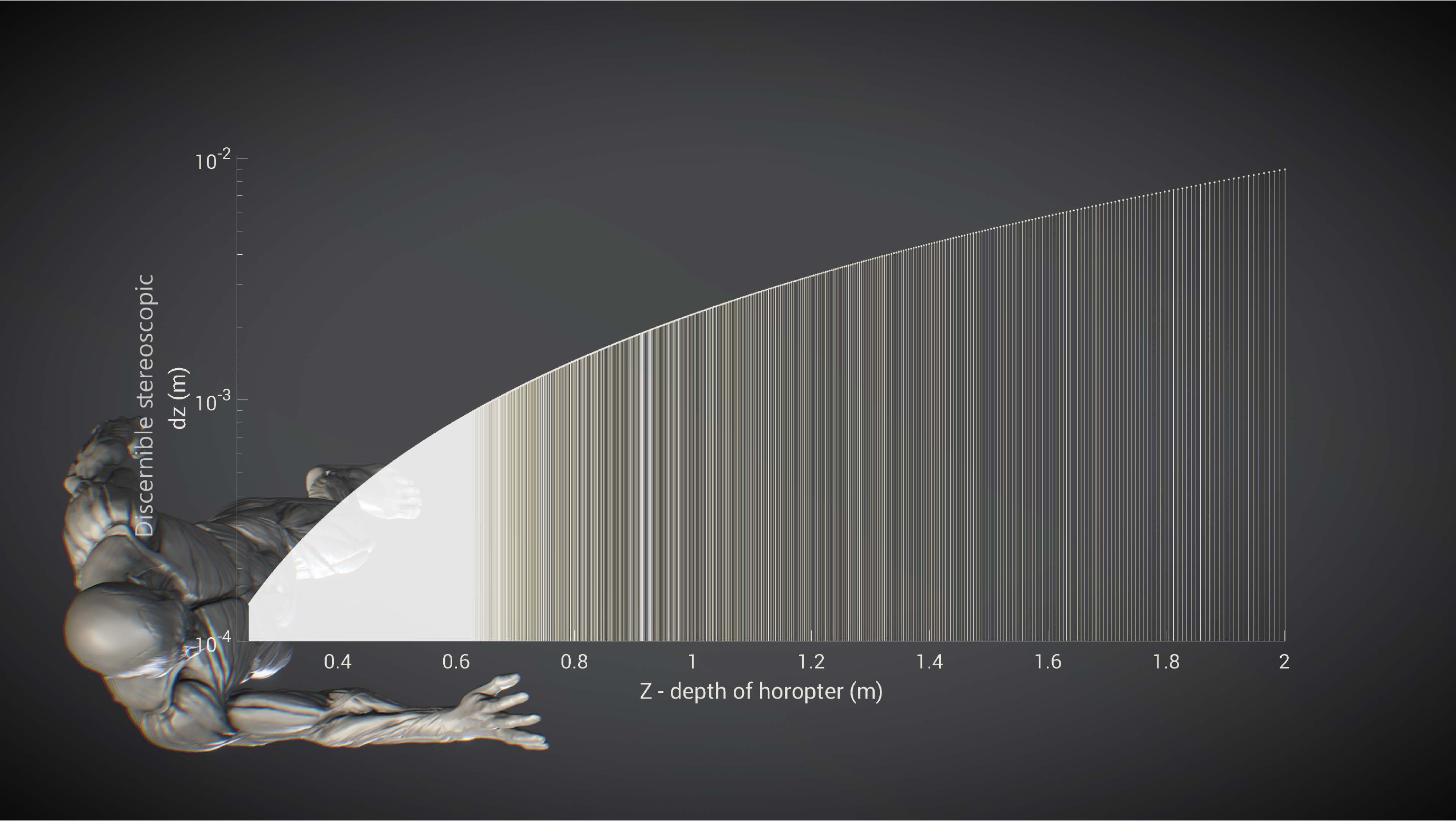}
		\put (48.5, -3) {\scalebox{.8}{(a)}} 
	\end{overpic}\\[.7cm]
	\centering \begin{overpic}[trim={0cm 0cm  0cm 0cm},clip,width=4.5in]{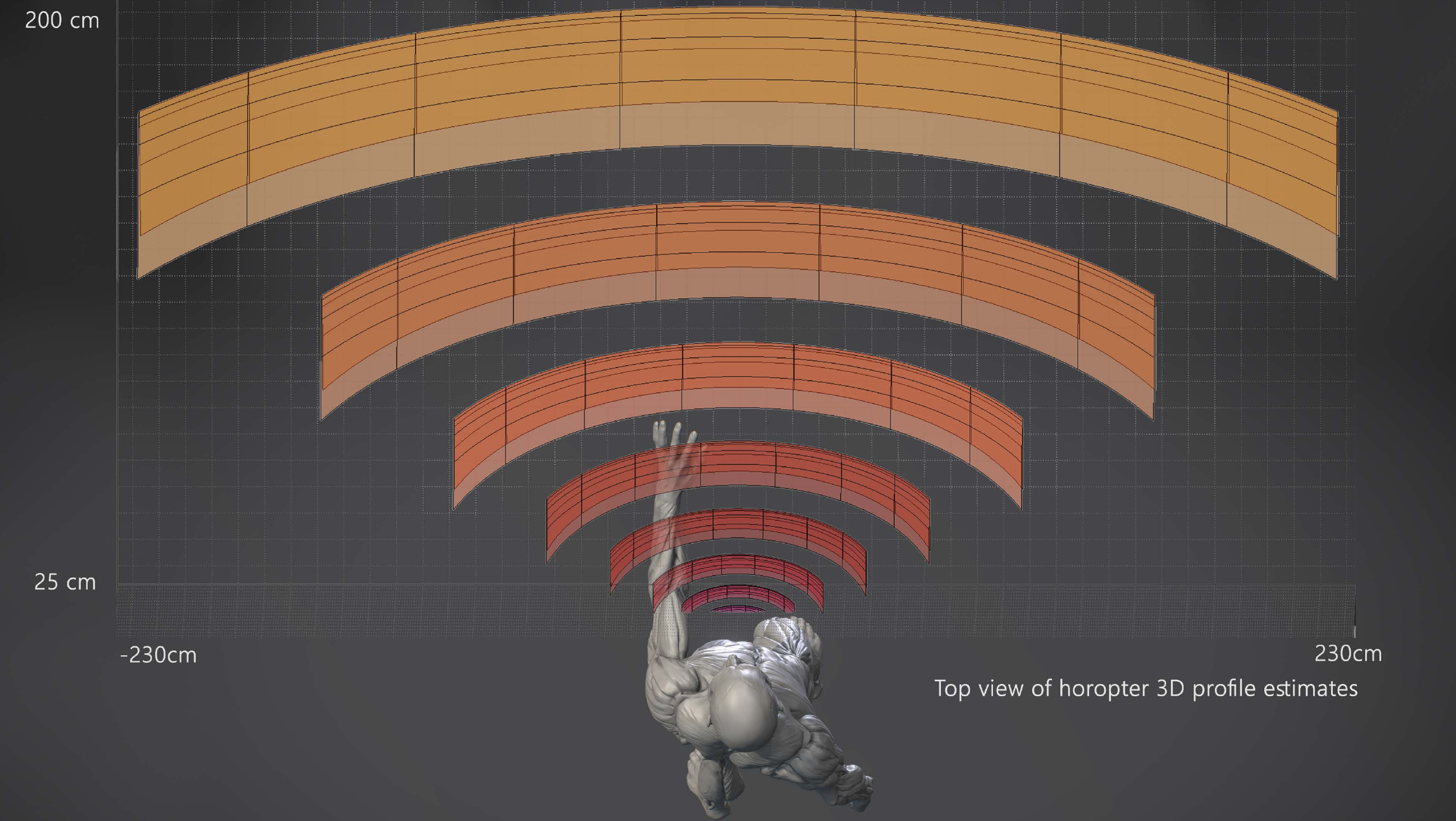}
		\put (48.5, -3) {\scalebox{.8}{(b)}} 
	\end{overpic}
	\vspace{.2cm}
	\caption{Quantized binocular depth level distances and schematic of spatial profile are shown from the face of an average 40-year old eye in one-to-one scale with regard to the male anatomy. (a): Discernible stereoscopic range at a given depth. Each line is one stereoscopic depth level. (b): Schematic view of the horopters from the top in one-to-one scale with regards to an average 40-year old male anatomy. Since horopters are extremely dense at closer distances, we have only shown a sparse set of these layers up to 2 meters. }\label{Supp:figGuy3}
\end{figure}

\subsection{Supplementary Note 4: Experimental validation guide}\label{Supp:ExpVal}
The experimental setup (Figure \ref{Supp:figExp})  overlays the images from two different displays and feeds them to the subject eye in Badal geometry (7.5 cm focal length of Badal lens would be suitable). The depth of each display is tunable independently. A photorefractometer reads the eye accommodation response and pupil size via a hot mirror. This setup allows us to measure subjectively reported defocus. Subjective reports are more tolerate to defocus, and are of notable practical value in order to design a light field optical system that can be commercialized. In order to achieve this, we do not need to paralyze the subject’s eye, and will use the method of limits to probe the threshold to detect blur, which we can define as the upper limit or lower limit of the depth of field.  At the same time, we can record the objective accommodation power by a photo refractometer that reads the refraction of the subject’s eye. The Badal optical system allows the experiment conductor and subject to change the target stimulus distance (change the accommodation) without changing its apparent size.

\begin{figure}[htb!]
	\centering \begin{overpic}[trim={0cm 0cm  0cm 0cm},clip,width=5in]{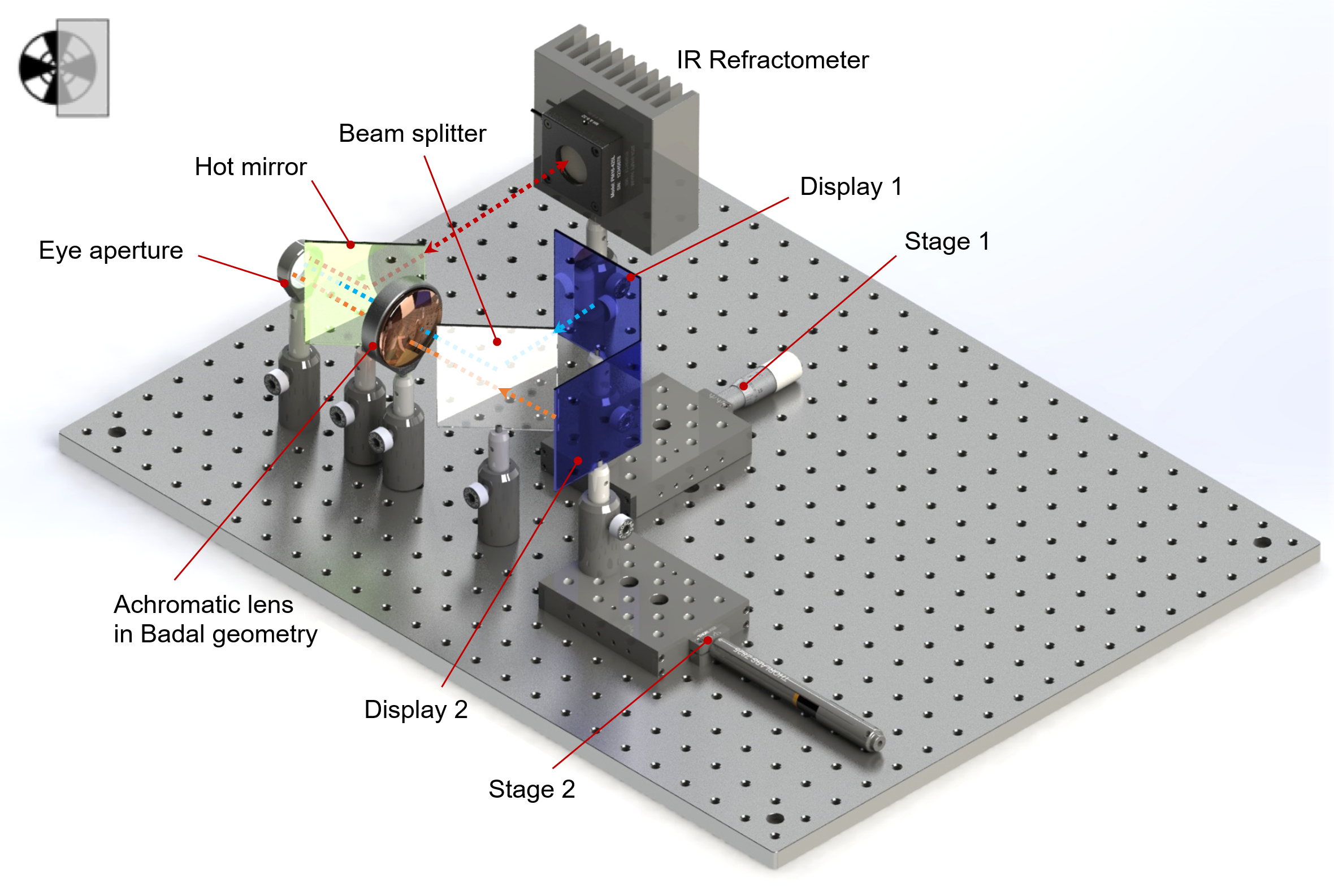}
	\end{overpic}
	\vspace{.1cm}
	\caption{Suggested experimental setup for measuring the quantized monocular depth levels response in subjects using Badal lens. The two images are shifted next to each other at foveal region until the subject indicates that there is a difference perceived in depth.   }\label{Supp:figExp}
\end{figure}

The stimulus should have notable varying spatial contrast (see inset image on the top left of Figure \ref{Supp:figExp}). Subjects will be instructed to accommodate on the left hemifield at all times. While the optical distance of the left hemifield will be fixed, the optical distance of the right hemifield can be moved by adjusting the lens power in the Badal system. Finally, the subject will be instructed to fixate on the left hemifield, and maintain their focus. This is done monocularly. We can use a photo refractometer to monitor their accommodation to make sure their focus did not drift. Then the optical distance of the right hemifield will be moved slowly either closer or further, until the subject reports by a key press at the point of the initial sight loss in clarity of the target. 

The initial distance of both targets will be varied across trials. The range of the distances will be decided based on a pilot study that tests a few distances. The pilot study will cover a range from 15 cm to infinity (6-0.09 D), as they overlap with the distances used in the theoretical simulation.

Since the pupil diameter (variation induced by luminance level in the environment), contrast, and spatial frequency of the stimulus have a big impact on the contrast sensitivity, we also would like to vary these variables and investigate how they contribute to the depth of field in future studies. While some literature showed less sensitivity to small amounts of retinal defocus in aging population, we would like to conduct this experiment on young healthy eyes with normal or corrected to normal visual acuity in a large range of distances.

\subsection{Supplementary Note 5: MATLAB Implementation}\label{Supp:secImplementation}
The code to implement the optimal DoF allocation is made available in MATLAB, and can be  accessed at:\\
\url{https://github.com/aaghasi/DoF-Allocation}.\\
The link contains the main functions to perform the optimal allocation, along with an example demo. Our code uses Gurobi \cite{gurobi} as the MBP/LP solver, which comes with a free license for academic use. In the following we briefly overview the main functions and scripts that need to be executed one after the other to perform a complete allocation. 

-- \texttt{createDoFTrain:}
This script allows the user to create a train of DoF profiles as functions of the depth and age. The code uses some basic parameters such as the pupil diameter,  minimum diopter range, and parameters controlling the resolution of the train knolls, to perform this task. Running \texttt{createDoFTrain} would produce the variables \texttt{profileTrain} and \texttt{centers}, which correspond to the train of DoF profiles and their corresponding centers. These files are encapsulated  as \texttt{RawProfileTrains.mat}, and saved to be used in the subsequent steps. The code also produces a snapshot of the profile trains, similar to Figure \ref{figDoFs}(a). 

-- \texttt{hypoGen:} Once the DoF profile train is generated, this function can be called to generate the binary matrix $\bPi$. While the matrix $\bPi$ can be readily used in the MBP program, the condensing process discussed in the next section can significantly reduce its size and produce an equivalent matrix with significantly fewer rows. 

-- \texttt{LPMBP:} This is the function that solves the main MBP, and uses the function \texttt{Gurobi\_LP\_Solver} to call the Gurobi optimizer. The \texttt{LPMBP} function first tries to solve the LP relaxation (which is faster), and if the solution is not binary, it then tries to call the Gurobi integer programming routines. 

-- \texttt{Demo:} To present a complete example of DoF allocation, the code \texttt{Demo} makes use of the functions described above in order, and perform two sets of DoF allocation for $T=1,\ldots,9$ that produce a result similar to Figure \ref{figOpt}(b).

\subsubsection{Condensing the MBP Matrix}\label{Supp:Condense}
In the paper we discussed that after discretizing the domain into $p$ pixels (voxels), and introducing the variables $\beta_j$, program (\ref{eqT3}) can be cast as the program (\ref{eqT5}). Notice that the number of variables in the (\ref{eqT5}) program is $n+p$, where $n$ is the number of knolls and $p$ is the number of pixels (voxels). Furthermore, the matrix $\bPi$ would be of size $p\times n$. When a fine grid is used to discretize the domain, or our domain is multi-dimensional, the number of (\ref{eqT5}) variables and the number of rows in  $\bPi$ can become very large. In this section we discuss a way to effectively condense the matrix $\bPi$, and in practice work with significantly smaller optimization programs.

\begin{figure}[htb!]
	\centering \begin{overpic}[trim={0cm 0cm  0cm 0cm},clip,width=5in]{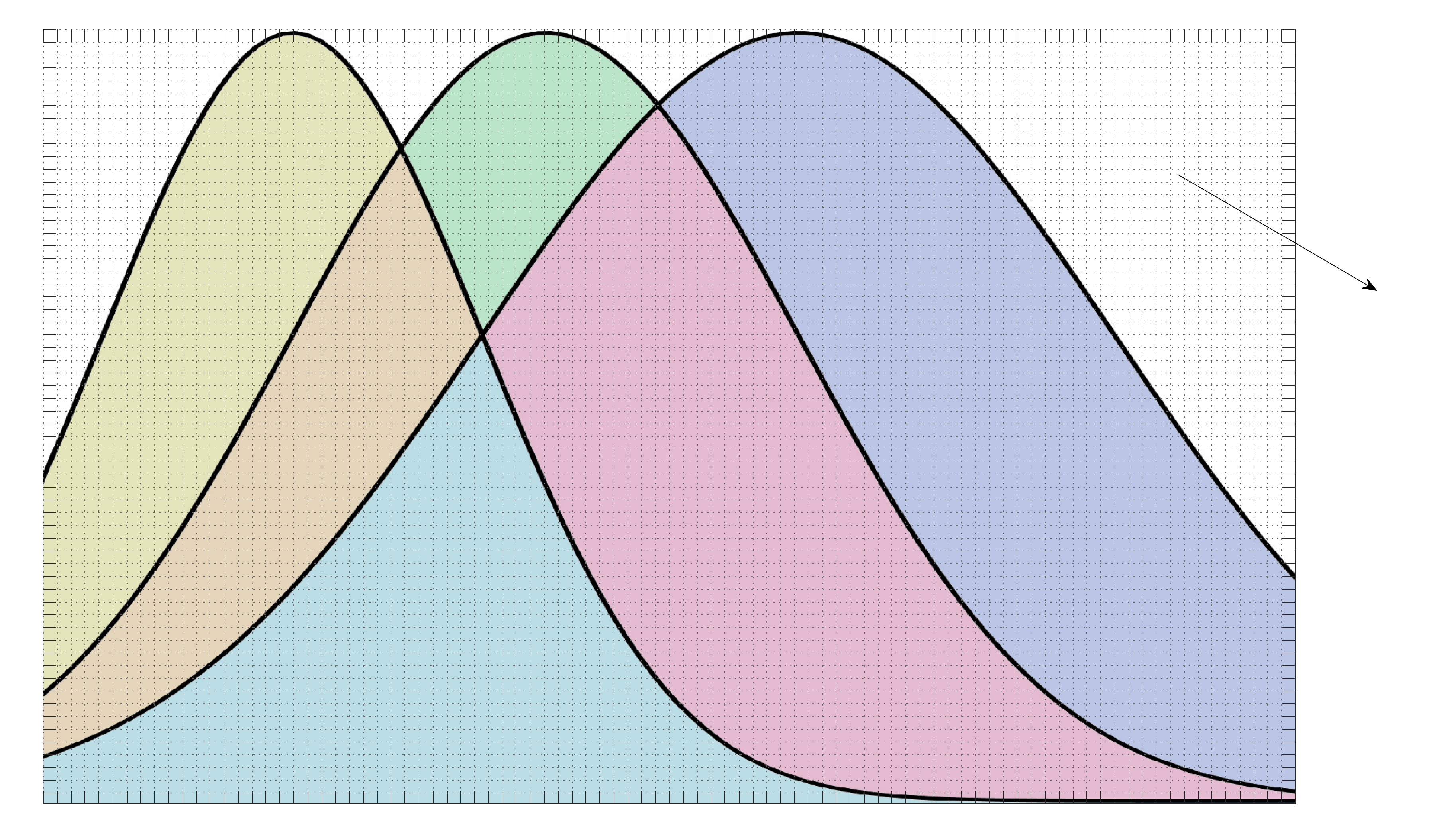}
		\put (97,36.5) {\scalebox{.8}{a pixel}} 
		\put (15.5,50) {\scalebox{1.2}{$\texttt{hyp}(f_1)$}} 
		\put (33,50) {\scalebox{1.2}{$\texttt{hyp}(f_2)$}} 
		\put (52,50) {\scalebox{1.2}{$\texttt{hyp}(f_3)$}} 
		\put (10,30) {\scalebox{1.2}{$\Omega_1$}} 
		\put (23,30) {\scalebox{1.2}{$\Omega_2$}} 
		\put (30,15) {\scalebox{1.2}{$\Omega_3$}} 
		\put (33,44) {\scalebox{1.2}{$\Omega_4$}} 
		\put (46,30) {\scalebox{1.2}{$\Omega_5$}} 
		\put (65,30) {\scalebox{1.2}{$\Omega_6$}} 
		
	\end{overpic}
	\caption{The process of condensing the matrix $\bPi$, forms an equivalent $\bPi_c$ matrix with only 6 rows, instead of $5400$ rows (which corresponds to the number of pixels in the domain).  }\label{Supp:figCond}
\end{figure}
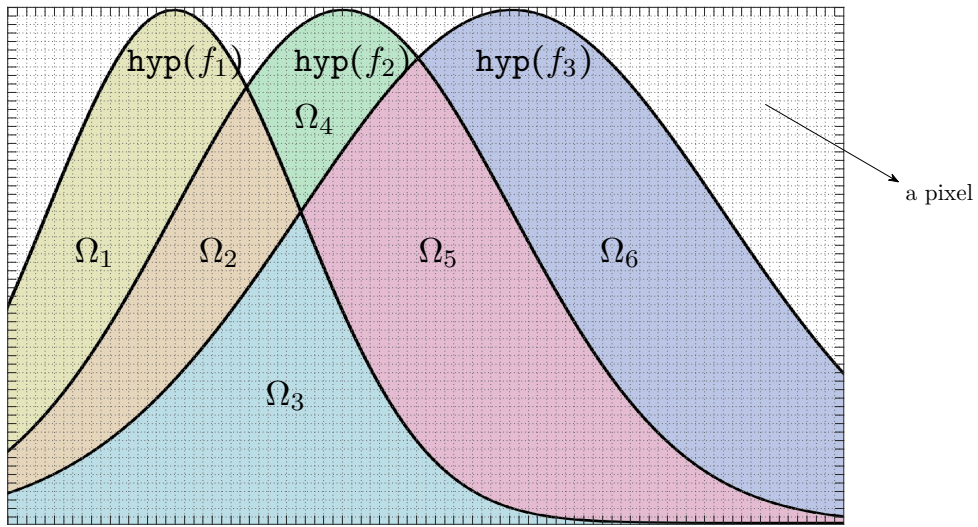
To visually explain the idea, consider a small problem with $n=3$ knolls, as presented in Figure \ref{Supp:figCond}.  Each knoll is in the form of a Gaussian. If the domain is discretized into a grid of $60\times 90$ (i.e., $p=5400$), the matrix $\bPi$ would be of size $5400\times 3$. 

From the figure, one could immediately see that the area under the knolls can be characterized as
\[\bigcup_{i=1}^3\texttt{hyp}(f_i)=\bigcup_{j=1}^6\Omega_j,
\]
where $\Omega_j$ are the \emph{disjoint} regions indicated in the figure. It is interesting to note that the regions $\Omega_j$ can be precisely characterized by the hypograph of the knolls. For example $\Omega_1$ belongs to $\texttt{hyp}(f_1)$, but does not belong to $\texttt{hyp}(f_2)$ or $\texttt{hyp}(f_3)$, and hence can be characterized as $\texttt{hyp}(f_1)\cap \texttt{hyp}(f_2)^c\cap \texttt{hyp}(f_3)^c$. Similarly, $\Omega_2$ can be written as $\texttt{hyp}(f_1)\cap \texttt{hyp}(f_2)\cap \texttt{hyp}(f_3)^c$, and so on. In \cite{aghasi2015convexa}, these regions are referred to as shapelets.

Thanks to the non-overlapping nature of these regions, for all the points $(x,y)$ in each $\Omega_j$, the value of the sum  $\sum_{i=1}^n \alpha_i\pi_{f_i}(x,y)$ is identical. For example, the value of $\sum_{i=1}^3 \alpha_i\pi_{f_i}(x,y)$ is equal to $\alpha_1$ for all the points in $\Omega_1$ (simply because $\pi_{f_2}(x,y)$ and $\pi_{f_3}(x,y)$ are zero in $\Omega_1$), and is $\alpha_1+\alpha_2$ for all the points in $\Omega_2$. Also, in the matrix $\bPi$, the rows corresponding to the pixels within each $\Omega_j$ are identical. For example, the row corresponding to every pixel that resides within $\Omega_1$ is $[1,0,0]$. Similarly, the row corresponding to every pixels that resides within $\Omega_2$ is $[1,1,0]$.

Based on this observation, one may use the super-pixels $\Omega_j$, instead of the fine pixels defined by the grid. This process can be performed by converting the matrix $\bPi$ into a condensed matrix $\bPi_c$. All that is needed is to add up all the identical rows in $\bPi$, and represent them with a single row in $\bPi_c$. For example, below, the condensed version of the matrix on the left is the matrix on the right:
\[\begin{pmatrix} 1 & 1 & 0\\ 0 & 1 & 0\\ 0 & 1 & 0\\ 1 & 1 & 0\\ 1 & 1 & 0\\ 1& 1&1\end{pmatrix}~~~ \underset{\xrightarrow{\hspace*{1cm}}}{condense} ~~ \begin{pmatrix} 3 & 3 & 0\\ 0 & 2 & 0\\1& 1&1\end{pmatrix}.
\]
Consider $p_c$ to be the number of rows in the condensed matrix $\bPi_c$, which turns out to be the number of super-pixels $\Omega_j$.  One can solve the  (\ref{eqT5}) program by using $\bPi_c$ instead of $\bPi$, and deal with a smaller optimization of size $n+p_c$. The new optimization would be in the form
\begin{align}
\minimize_{\balpha,\bbeta }&~ - \boldsymbol{1}^\top \bbeta ~~ ~ \mbox{subject to:}~ \left\{\begin{array}{l}\bbeta\leq \bPi_c\balpha ,  ~~~ \boldsymbol{1}^\top \balpha\leq T  \\ \bbeta\leq \vu,~~~~~~~ \alpha_i\in\{0,1\} ~~ i\in[n] \end{array} \right.. \label{reducedOpt}
\end{align}
Previously, the bound constraint on $\bbeta$ was $\bbeta\leq \boldsymbol{1}$, which in the condensed formulation needs to modify to $\bbeta\leq \vu$, where the $j$-th element of $\vu$ is the number of pixels within the super-pixel $\Omega_j$. Thanks to the properties stated above for the super-pixels, the $\balpha$ part of both optimization  solutions (which is the part we seek) would be identical. 

In the Matlab package, the function \texttt{PiCondenser} performs the condensing procedure described above. The function is used as
\begin{verbatim}
[Pi_c, Ncount] = PiCondenser(Pi);
\end{verbatim}
where the matrix $\bPi$ (shown as \texttt{Pi}) is passed as the input, and the program calculates the matrix $\bPi_c$ (shown as \texttt{Pi\_c}) and the vector $\vu$ ((shown as \texttt{Ncount}) used in \eqref{reducedOpt}.



\end{document}